\documentclass[runningheads]{llncs}

\usepackage{booktabs,comment}
\usepackage{colortbl}
\usepackage[ruled]{algorithm2e}
\usepackage[mathscr]{eucal}
\usepackage{tikz}
\usetikzlibrary{arrows, automata,shapes}
\usetikzlibrary{decorations.pathreplacing}
\usepackage{hyperref,multirow,xcolor}
\usepackage{comment,amsmath,amssymb,amstext,amsfonts,graphics, enumerate}
\usepackage{color,graphicx,latexsym,amsfonts,epsfig,eufrak,array,boxedminipage,graphicx,bm}

\SetAlFnt{\small}
\SetAlCapFnt{\small}
\SetAlCapNameFnt{\small}
\SetAlCapHSkip{0pt}
\IncMargin{-\parindent}

\newcommand \dia{\hfill{\textcolor{darkgray}{$\diamond$}}}

\newcommand{\N}{\mathbb{N}}
\newcommand{\R}{\mathbb{R}}
\newcommand{\NP}{{\sf NP}}

\begin{document}
\title{Partitioned Matching Games for International Kidney~Exchange
\thanks{Some results in this paper appeared in the proceedings of AAMAS 2019~\cite{BKPP19} and AAMAS 2022~\cite{BBKP22}.
We dedicate this paper to the memory of Walter Kern, who passed away in 2021. Benedek was supported by the National Research, Development and Innovation Office of Hungary (OTKA Grant No.\ K138945). Bir\'o was supported by the Hungarian Scientific Research Fund (OTKA, Grant No.\ K143858) and
 by the Hungarian Academy of Sciences (Momentum Grant No. LP2021-2).}}

\author{M\'arton Benedek\inst{1,2} \and
P\'eter Bir\'o\inst{1,2} \and
Walter Kern\inst{3} \and
D\"om\"ot\"or P\'alv\"olgyi\inst{4}\and
Daniel Paulusma\inst{5}}

\authorrunning{M. Benedek et al.}
\institute{KRTK, Budapest, Hungary \email{\{peter.biro,marton.benedek\}@krtk.hu}\and
Corvinus University of Budapest, Budapest, Hungary \and University of Twente, Enschede, The Netherlands \email{w.kern@utwente.nl} \and
MTA-ELTE Lend\"ulet, Budapest, Hungary \email{domotorp@gmail.com}
\and
Durham University, Durham, UK \email{daniel.paulusma@durham.ac.uk}}

\maketitle

\begin{abstract}
We introduce partitioned matching games as a suitable model for international kidney exchange programmes, where in each round the total number of available kidney transplants needs to be distributed amongst the participating countries in a ``fair'' way.
 A partitioned matching game $(N,v)$ is defined on a graph $G=(V,E)$ with an edge weighting $w$ and a partition $V=V_1 \cup \dots \cup V_n$. The player set is $N = \{ 1, \dots, n\}$, and player $p \in N$ owns the vertices in $V_p$. The value $v(S)$ of a coalition~$S \subseteq N$ is the maximum weight of a matching in the subgraph of $G$  induced by the vertices owned by the players in~$S$. If $|V_p|=1$ for all $p\in N$, then we obtain the classical matching game.
Let $c=\max\{|V_p| \; |\; 1\leq p\leq n\}$
 be the width of $(N,v)$.
We prove that checking core non-emptiness is polynomial-time solvable if $c\leq 2$ but co-\NP-hard if $c\leq 3$. We do this via pinpointing a relationship with the known class of $b$-matching games and completing the complexity classification on testing core non-emptiness for $b$-matching games.
With respect to our application, we prove a number of complexity results on choosing, out of possibly many optimal solutions, one that leads to a kidney transplant distribution that is as close as possible to some prescribed fair distribution.

\end{abstract}

\setcounter{footnote}{0}
\section{Introduction}\label{s-intro}

We consider two generalizations of a classical class of games in cooperative game theory, namely matching games, which in turn generalize the well-known class of assignment games. One of the two generalizations of matching games is the known class of $b$-matching games. The other one, the class of partitioned matching games, is a new class of cooperative games introduced in this paper. We show how these two generalizations are related to each other and justify the new class of cooperative games by an emerging real-world application: international kidney exchange~\cite{Bo_etal17,Va_etal19}. Partitioned matching games implicitly played a role in international kidney exchange through the work of  \cite{BGKPPV20,KNPV20} and have recently been used as a basis for simulations in~\cite{BBKP22}. The goal of this paper is to provide a strong {\it theoretical} basis for partitioned matching games by proving a number of computational complexity results on computing core allocations and finding allocations close to prescribed ``fair'' (core) allocations. We start with introducing the necessary basic terminology.

\subsection{Basic Terminology}\label{s-basic}

A \emph{(cooperative) game} is a pair $(N,v)$, where $N$ is a set of $n$ \emph{players} (agents) and $v: 2^N\to \R$
 is a \emph{value function} with $v(\emptyset) = 0$.  A subset $S\subseteq N$ is a {\it coalition} and the set $N$ is the {\it grand} coalition. For many natural games, it holds that $v(N)\geq v(S_1)+\cdots +v(S_r)$ for every possible partition $(S_1,\ldots,S_r)$ of $N$. Hence, the central problem is how to keep the grand coalition $N$ stable by distributing $v(N)$ amongst the players of $N$ in a ``fair'' way. Such distributions of $v(N)$ are also called allocations.
That is, an {\it allocation} for a game $(N,v)$ is a vector $x \in \R^N$ with $x(N) = v(N)$; here, we write $x(S)=\sum_{p\in S}x_p$ for a subset $S\subseteq N$.
A {\it solution concept} prescribes a set of fair allocations for a game $(N,v)$, where the notion of fairness depends on context.

One of the best-known solution concepts, the {\it core} of a game consists of all allocations $x \in \R^N$ satisfying $x(S)\geq v(S)$ for each $S\subseteq N$. Core allocations are highly desirable, as they offer no incentive for a subset $S$ of players to leave~$N$ and form a coalition on their own. So core allocations ensure that the grand coalition~$N$ is stable. However, the core may be empty. Moreover, the following three problems may be computationally hard
for a class of cooperative games (assuming a ``compact'' description of the~input game, which is often, and also in our paper, a graph with an edge weighting):

\begin{itemize}
\item [{\bf P1.}] determine if a given allocation $x$ belongs to the core, or find a coalition $S$ with $x(S)<v(S)$,
\item [{\bf P2.}] determine if the core is non-empty, and
\item [{\bf P3.}] find an allocation in the core (if it is non-empty).
\end{itemize}
If P1 is polynomial-time solvable for some class of games, then so are P2 and P3 using the ellipsoid method~\cite{GLS81,Ka79}.
As the core of a game might be empty, other solution concepts are also considered. Well-known examples of other solution concepts are the Shapley value and nucleolus (which are both computationally hard to compute).

The input games we consider are defined on an undirected graph~$G=(V,E)$ with a positive~edge weighting $w:E\to \R_+$; here, $V$ is a set of vertices and $E$ is a set of edges between pairs of distinct vertices. Such a game is said to be {\it uniform} (or {\it simple}) if $w\equiv 1$.
For a subset $S\subseteq V$, we let $G[S]$ denote the subgraph of $G$ {\it induced} by $S$, that is, $G[S]=(S,\{uv\in E\; |\; u,v\in S\})$ is the graph obtained from~$G$ after deleting all the vertices outside $S$. A graph is {\it bipartite} if its vertex set can be partitioned into two sets $A$ and~$B$ such that every edge joins a vertex in $A$ to a vertex in $B$.
A {\it matching} $M$ is a set of edges in a graph~$G$ such that no two edges of $M$ have a common end-vertex.

A {\it matching game} defined on a graph~$G=(V,E)$ with a positive edge weighting $w:E\to \R_+$ is the game $(N,v)$
where $N=V$ and the value $v(S)$ of a coalition $S$ is equal to the maximum weight $w(M)=\sum_{e\in M}w(e)$ over all matchings $M$ of $G[S]$; see Figure~\ref{f-first} for an example. Matching games defined on bipartite graphs are commonly known as {\it assignment games}.

\begin{figure}[t]
\centering
\begin{tikzpicture}[scale=0.9,rotate=0]
\draw
(0, 1) node[circle, black, scale=0.8,draw](a){\small$1$}
(1.5, 2) node[circle, black, scale=0.8,draw](b){\small$2$}
(1.5, 0) node[circle, black, scale=0.8,draw](c){\small$3$}
(3, 3) node[circle, black, scale=0.8,draw](d){\small$4$}
(3, 1) node[circle, black, scale=0.8,draw](e){\small$5$}
(3, -1) node[circle, black, scale=0.8,draw](f){\small$6$};
\draw[-] (a) -- node[above] {\small$2$} (b);
\draw[-] (b) -- node[above] {\small$2$} (d);
\draw[-] (d) -- node[right] {\small$3$} (e);
\draw[-] (b) -- node[above] {\small$1$} (e);
\draw[-] (a) -- node[below] {\small$1$} (c);
\draw[-] (c) -- node[above] {\small$3$} (e);
\draw[-] (c) -- node[below] {\small$2$} (f);
\draw[-] (e) -- node[right] {\small$2$} (f);
\end{tikzpicture}
\caption{An example~\cite{BBJPX} of a matching game $(N,v)$ on a graph $G=(V,E)$, so $N=V$. Note that $v(N)=7$ and that the core of $(N,v)$ is nonempty, e.g. the allocation $x=(\frac{1}{2}, \frac{3}{2}, \frac{3}{2}, 1, 2, \frac{1}{2})$ belongs to the core.}\label{f-first}
\end{figure}
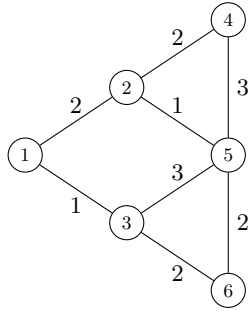

Let $G=(V,E)$ be a graph with a vertex function $b:V\to \N_+$, which we call a {\it capacity} function.
A {\it $b$-matching} of $G$ is a subset $M \subseteq E$ such that each vertex $p \in V$ is incident to at most $b(p)$ edges in $M$.
Let $w:E \to \R_+$ be an edge weighting of $G$.
A {\it $b$-matching game} defined on $(G,b,w)$ is the game $(N,v)$ where $N=V$
and the value $v(S)$ of a coalition $S$ is the maximum weight $w(M)$ over all $b$-matchings~$M$ of $G[S]$.
Note that if $b\equiv 1$, then we obtain a matching game.
If $G$ is bipartite, then we speak of a {\it $b$-assignment game}.

In this paper we introduce the notion of a partitioned matching game.
Let $G=(V,E)$ be a graph with an edge weighting $w:E\to \R_+$. Let ${\mathcal V}=(V_1,\ldots,V_n)$ be a partition of $V$.
A {\it partitioned matching game} defined on $(G,w,{\cal V})$ is  the game $(N,v)$ where $N=\{1,\ldots,n\}$
and the value $v(S)$ of a coalition $S$ is the maximum weight $w(M)$ over all matchings~$M$ of
 $G[\bigcup_{p\in S}V_p]$.
Let $c=\max\{|V_p| \; |\; 1\leq p\leq n\}$ be the {\it width} of $(N,v)$.
Note that if $V_p=\{p\}$ for every $p\in N$, then we obtain a matching game.

We observe that if $(N,v)$ is a $b$-matching game or a partitioned matching game, then $v(N)\geq v(S_1)+\cdots +v(S_r)$ for every possible partition $(S_1,\ldots,S_r)$ of~$N$. Hence,
studying problems P1--P3 for both $b$-matching games and partitioned matching games is meaningful.

Before presenting our new results on P1--P3 for $b$-matching games and partitioned matching games in Section~\ref{s-new}, we first
survey the known results on P1--P3 for $b$-assignment games and $b$-matching games, including results for $b\equiv 1$, in Section~\ref{s-known}.
Afterwards, we describe the setting of international kidney exchange as an application of partitioned matching games in Section~\ref{s-app}.

\subsection{Known Results for P1--P3 for Matching Games and Their Variants}\label{s-known}

It follows from results of Koopmans and Beckmann~\cite{KB57} and Shapley and Shubik~\cite{SS72} that
the core of every assignment game is non-empty. In fact, this holds even for $b$-assignments for any capacity function~$b$, as shown by Sotomayor~\cite{So92}.
However, the core of a matching game might be empty (for example, if $G$ is a triangle with $w\equiv 1$). Problem~P1 is linear-time solvable for matching games: the problem is equivalent to verifying whether for an allocation~$x$, it holds that $x_p+x_q\geq w(pq)$ for every edge $pq\in E$. Hence, problems P2 and P3 are polynomial-time solvable for matching games (and thus also for assignment games). We refer to~\cite{BKP12,DIN99} for some alternative polynomial-time algorithms for solving P2 and P3 on matching games.

Sotomayor~\cite{So92} proved that P3 is polynomial-time solvable for $b$-assignment games. However,
Bir\'o et al.~\cite{BKPW18} proved that P1 is co-\NP-complete even for uniform $3$-assignment games.
They also showed that  for $b\leq 2$, P1 (and so, P2 and P3) is polynomial-time solvable for $b$-matching games.

Table~\ref{t-tabtab} summarizes the known results for problems P1--P3 for matching games and their variants (in this table we have also included our new results, which we explain below in detail).
In Section~\ref{s-con} we discuss future work. There, we also mention some results for other solution concepts.
For an in-depth discussion on complexity aspects of solution concepts for matching games and their variants,
we refer to the recent survey of Benedek et al.~\cite{BBJPX}.

\begin{table}[t]
\centering
\begin{tabular}{r|c|c|c}
     & problem P1$\;$ & problem P2$\;$     & problem P3$\;$\\
\hline
assignment games $\;\;\;\;\;\;\;\;\;\;\;\;$ & poly &yes &poly\\
\hline
matching games $\;\;\;\;\;\;\;\;\;\;\;\;$  &poly &poly &poly\\
\hline
$b$-assignment games:
if $b\leq 2$   & poly  &yes &poly\\
if $b\leq 3$ 			 &co\NP c &yes &poly\\		
\hline
$b$-matching games: if $b\leq 2$  & poly  &poly      & poly      \\
if $b\leq 3$          &co\NP c 			 &{\bf co\NP h}  & {\bf co\NP h} \\
\hline
{partitioned matching games:}
{if $c\leq 2$}  & {\bf poly}  &{\bf poly}      & {\bf poly}      \\
{if $c\leq 3$}             &{\bf co\NP c} 			 &{\bf co\NP h}  & {\bf co\NP h} 			
\end{tabular}
\vspace*{2mm}
\caption{Complexity dichotomies for the core (P1--P3), with the following short-hand notations,
yes: all instances are yes-instances; poly: polynomial-time; co\NP c: co-\NP-complete; and co\NP h: co-\NP-hard.
{Recall that $b:V\to \N_+$ is a function and that $c=\max\{|V_p| \; |\; 1\leq p\leq n\}$ is a natural number.}
The new results of this paper are put in {\bf bold}.
The seven hardness results in the table hold even if $w\equiv 1$.}\label{t-tabtab}
\end{table}

\subsection{New Results for P1--P3 for the Two Generalizations of Matching Games}\label{s-new}

We start with the known generalization of matching games, namely the $b$-matching games. In Section~\ref{s-b} we prove the following result, solving the only case left open in~\cite{BKPW18} (see also Table~\ref{t-tabtab}).

\begin{theorem}\label{t-main1}
\emph{P2} and \emph{P3} are co-\NP-hard for uniform $b$-matching games if $b\leq 3$.
\end{theorem}

\noindent
For a capacity function $b:V\to \N_+$, we let $b^*$ denote the maximum $b(p)$-value over all $p\in V$.
In Section~\ref{s-b} we also prove the following two theorems, which together show that there exists a close relationship between $b$-matching games and partitioned matching games, our new generalization of matching games.

\begin{theorem}\label{t-struc1}
\emph{P1}--\emph{P3} can be reduced in polynomial time from $b$-matching games to partitioned matching games of width $c=b^*$, preserving uniformity.
\end{theorem}

\begin{theorem}\label{t-struc2}
\emph{P1}--\emph{P3} can be reduced in polynomial time from partitioned matching games of width~$c$ to $b$-matching games for some capacity function $b$ with $b^*\leq c$.
\end{theorem}

\noindent
Combining Theorems~\ref{t-main1}--\ref{t-struc2} with the aforementioned results of Bir\'o et al.~\cite{BKPW18}, which state that P1--P3 are polynomial-time solvable for $b$-matching games if $b\leq 2$ and that P1 is co-\NP-complete even for uniform $b$-assignment games if $b\leq 3$, leads to the following dichotomy (see also Table~\ref{t-tabtab}).

\begin{corollary}\label{c-main}
For $c\leq 2$, \emph{P1}--\emph{P3} are polynomial-time solvable for partitioned matching games with width~$c$, whereas for $c\leq 3$,
\emph{P1} is co-\NP-complete, and \emph{P2}--\emph{P3} are co-\NP-hard for uniform partitioned matching games.
\end{corollary}

\subsection{An Application of Partitioned Matching Games}\label{s-app}

As a strong motivation for introducing partitioned matching games, we consider international kidney exchange programmes. We explain these programmes below but note that partitioned matching games can also be used to model other economic settings where multi-organizations control pools of clients~\cite{GMP13}.

The most effective treatment for kidney patients is transplantation. The best long-term outcome is to use {\it living} donors. A patient may have a willing donor, but the donor's kidney will often be rejected by the patient's body if donor and patient are medically incompatible. Therefore, in the last 30 years, an increasing number of countries have started to run a (national) Kidney Exchange Programme (KEP)~\cite{Bi_etal2019}. In such a programme, a patient and their incompatible donor are placed in a pool with other patient-donor pairs.  If donor~$d$ of patient~$p$ is compatible with patient~$p'$ and simultaneously, donor~$d'$ of $p'$ is compatible with $p$, then the pairs $(p,d)$ and $(p',d')$ are said to be {\it compatible} and a {\it $(2$-way) exchange} between $(p,d)$ and $(p',d')$ can take place. That is, the kidney of donor $d$ can be given to patient $p'$, and the kidney of donor $d'$ can be given to patient $p$.

KEPs operate in regular rounds, each time with an updated pool where some patient-donor pairs may have left and new ones may have entered.
Naturally, the goal of a KEP is to help in each round as many patients as possible. For this goal, we construct for each round, the corresponding (undirected) {\it compatibility graph}~$G$ as follows. First, we introduce a vertex~$i^p_d$ for each patient-donor pair $(p,d)$ in the pool. Next, we add an edge between two vertices $i$ and $j$ if and only if the corresponding patient-donor pairs are compatible. We give each edge $ij$ of $G$ a weight $w(ij)$ to express the utility of an exchange between $i$ and $j$; we set $w\equiv 1$ if we do not want to distinguish between the utilities of any two exchanges.

It now remains to find a maximum weight matching of $G$, as such a matching corresponds to an optimal set of exchanges or, if we had set $w\equiv 1$, to a largest set of exchanges in the patient-donor pool. As we can find a maximum weight matching of a graph in polynomial time, we can find an optimal solution for each round of a KEP in polynomial time.

We can generalize $2$-way exchanges by defining the {\it directed compatibility graph} $\overline{G}=(V,A)$, in which $V$ consists of the patient-donor pairs and $A$ consists of every arc $(i,j)$ such that the donor of pair~$i$ is compatible with the patient of pair~$j$. Each arc $(i,j)\in A$ may have an associated weight~$w_{ij}$ expressing the utility of a kidney transplant involving the donor of $i$ and the patient of~$j$; note that $w_{ij}\neq w_{ji}$ is possible, should both arcs $(i,j)$ and $(j,i)$ exist in $\overline{G}$. For $\ell\geq 2$, an {\it $\ell$-way exchange} corresponds to a directed $\ell$-vertex cycle in a directed {compatibility} graph. It involves $\ell$ distinct patient-donor pairs $(p_1,d_1),\ldots, (p_\ell,d_\ell)$, where for $i\in \{1,\ldots,\ell-1\}$, donor $d_i$ donates to patient $p_{i+1}$ and donor $d_{\ell}$ donates to patient $p_1$. We refer to Figure~\ref{f-stexample} for an example. The same figure also illustrates how we can obtain an undirected compatibility graph $G=(V,E)$ from a directed compatibility graph $\overline{G}=(V,A)$: we add an edge $ij$ to $E$ if and only if both $(i,j)$ and $(j,i)$ are arcs in $A$ and in that case we set $$w(ij)=w_{ij}+w_{ji}.$$
Allowing $\ell$-way exchanges for some $\ell>2$ leads to possibly more patients being treated. However, Abraham, Blum and Sandholm~\cite{ABS07} proved that already for $\ell=3$ it becomes \NP-hard to find an optimal solution
 for a round. We therefore set $\ell=2$\footnote{Some countries allow $\ell=3$ or $\ell=4$, but in practice choosing $\ell=2$ is not uncommon due to lower risk levels (kidney transplants in a cycle must take place simultaneously)~\cite{Bi_etal2019}.},
 just like some related papers~\cite{CL19,CLPV17,STW21}, which we discuss later.

\begin{figure}[t]
\begin{center}
\begin{tabular}{cc}
 \begin{tikzpicture}[
            > = stealth,
            shorten > = 1pt,
            auto,
            node distance = 2.2cm,
            semithick
        ]

        \tikzstyle{every state}=[scale=0.7,
            draw = black,
            thick,
            fill = white,
            minimum size = 3mm
        ]
        \node[state] (i2) {$i_2$};
        \node[state] (i1) [above right of=i2] {$i_1$};
        \node[state] (j2) [right of=i2] {$j_2$};
        \node[state] (j1) [below right of=i2] {$j_1$};
        \node[state] (i3) [right of=j2] {$i_3$};
 \path[->]  (i1) edge node [right] {1} (i2);
        \path[->] [bend left] (i2) edge node {1} (i1);
        \path[->]  (i2) edge node [above] {1} (j2);
        \path[->] [bend left] (j2) edge node [below] {3} (i2);
        \path[->] (i2) [bend right] edge node [left] {2} (j1);
        \path[->] (j2) edge node [right] {2} (i1);
                \path[->] (j2) [bend left] edge node [right] {1} (j1);
        \path[->] (j1) edge node [left] {1} (j2);
        \path[->] (i1) edge node {5} (i3);
    \end{tikzpicture}
    \hspace*{0.5cm}
     \begin{tikzpicture}[
            > = stealth,
            shorten > = 1pt,
            auto,
            node distance = 2.2cm,
            semithick
        ]
\hspace*{2cm}
        \tikzstyle{every state}=[scale=0.7,
            draw = black,
            thick,
            fill = white,
            minimum size = 3mm
        ]

        \node[state] (i2) {$i_2$};
        \node[state] (i1) [above right of=i2] {$i_1$};
        \node[state] (j2) [right of=i2] {$j_2$};
        \node[state] (j1) [below right of=i2] {$j_1$};
        \node[state] (i3) [right of=j2] {$i_3$};
 \path[]  (i1) edge node [above left] {2} (i2);
        \path[]  (i2) edge node {4} (j2);
        \path[] (j1) edge node {2} (j2);
    \end{tikzpicture}
\end{tabular}
\end{center}
\caption{\label{fig:particular-graphs}{\it Left:} a directed {compatibility} graph $\overline{G}=(V,A)$ with a positive edge weighting~$w$ for a certain round of an international KEP; note that $\overline{G}$ has a directed $4$-vertex cycle, so in the corresponding KEP even a $4$-way exchange could take place.  {\it Right:} the corresponding undirected compatibility graph $G=(V,E)$ (which is used when only $2$-way exchanges are allowed).}\label{f-stexample}
\end{figure}
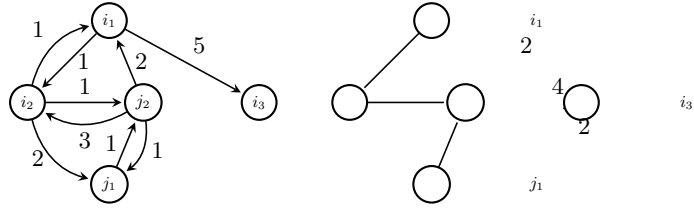

Nowadays, several countries are starting an {\it international} Kidney Exchange Programme (IKEP)  by merging the pools of their national KEPs; for example, Austria with the Czech Republic~\cite{Bo_etal17}; Denmark with Norway and Sweden; and Italy with Portugal and Spain~\cite{Va_etal19}. This may lead to more patients being helped.
However, apart from a number of ethical and legal issues, which are beyond the scope of this paper, we now have an additional issue that must be addressed.
Namely, in order to ensure full participation of the countries in the IKEP, it is crucial that {\it proposed solutions will be accepted by each of the participating countries}. Otherwise countries may decide
 to leave the IKEP at some point. So, in order for an IKEP to be successful, we need to ensure that the programme is {\it stable} on the long term. This is a highly non-trivial issue, as we can illustrate even with the following small example.

\medskip
\noindent
{\it Example 1.} Let $G$ be a compatibility graph with vertices $i_1,i_2,j$ and edges $i_1i_2$ and $i_2j$ with weights
$1-\epsilon$ and $1$ respectively, for some small $\epsilon$. The weight $1-\epsilon$ may have been obtained after desensitization, which we explain below.
Let $V_1=\{i_1,i_2\}$ and $V_2=\{j\}$.
The optimum solution is an exchange between $i_2$ and $j$ with weight $1$.
However, the solution consisting of the exchange between $i_1$ and $i_2$ (with weight~$1-\epsilon$) is better for $V_1$, as then both patients in the pairs $i_1$ and $i_2$ receive a kidney, and with more or less the same chance of success. \dia

\medskip
\noindent
To ensure stability, Carvalho and Lodi~\cite{CL19} used a $2$-round system with $2$-way exchanges only: in the first round each country selects an internal matching, and in the second round a maximum matching is selected for the international exchanges. They gave a polynomial-time algorithm for computing a Nash-equilibrium that maximizes the total number of transplants, improving the previously  known result of~\cite{CLPV17} for two countries. Sun et al. \cite{STW21} also considered $2$-way exchanges only. They defined so-called selection ratios using various lower and upper target numbers of kidney transplants for each country. In their setting, a solution is considered to be fair if the minimal ratio across all countries is maximized. They also required a solution to be a maximum matching and individually rational. They gave theoretical results specifying which models admit solutions with all three properties and they provided polynomial-time algorithms for computing such solutions, if they exist.

In contrast to~\cite{CL19,CLPV17,STW21}, we model an IKEP as follows. Consider a directed compatibility graph $\overline{G}=(V,A)$ with a positive arc weighting $w$. Let ${\cal V}=V_1\cup \cdots \cup V_n$ be a partition of $V$. Let $G=(V,E)$ be the corresponding undirected compatibility graph. If $w\equiv 1$, then we model a round of an IKEP  of a set $N$ of $n$ countries as a partitioned matching game $(N,v)$ defined on $(G,w)$ with partition~${\cal V}$. Else, we consider the {\it directed} partitioned matching game $(N,v)$ on $(\overline{G},w)$ with partition ${\cal V}$; so for each $S\subseteq N$, we have that $v(S)=\sum_{p\in S}\sum_{i,j: ij\in M, j\in V_p} w_{ij}$, where $M$ is a maximum weight matching in $G[S]$.
{If $c=1$, then we also use the term {\it directed matching game}.}
We define the following set:

$${\cal M}=\{M\; |\; M\; \mbox{is a maximum weight matching of}\; G\}.$$

\noindent
Note that ${\cal M}$ is the set of maximum matchings of $G$ if $w\equiv 1$ on $G$. We note also that ${\cal M}$ may consist of only a few matchings, or even a unique matching. The latter is highly likely
when weights $w(ij)$ take many different values at random~\cite{MVV87}. However, in our context, this will not be the case. Compatibility graphs will usually have only a {\it small} number of different weights. The reason is that to overcome certain blood and antigen incompatibilities, patients can undergo one or more desensitization treatments to match with their willing donor or some other potential donor.
 After full desensitization the chance on a successful kidney transplant is almost the same as in the case of full compatibility. Allowing desensitization results in compatibility graphs with weights either $1$ (when no desensitization was needed) or $1-\epsilon$ (after applying desensitization)~\cite{ACDM21}. Hence, in our application, ${\cal M}$ is likely to have {\it exponential size}.

Our approach is based on a credit-based system that was first introduced by Klimentova et al.~\cite{KNPV20}. In each round, we maximize the {\it total welfare}, that is, we choose some matching $M$ from the set ${\cal M}$ for that round. In addition, we specify a {\it target allocation~$x$}. A target allocation prescribes exactly which share of $v(N)$
each country ideally should receive in a certain round. However, if we use $M\in {\cal M}$, then the allocated share of $v(N)$ for country $p$ is
$$u_p(M)= \sum_{i,j: ij\in M, j\in V_p} w_{ij}.$$
Note that $\sum_{p\in N}u_p(M)=w(M)=v(N)$.
 The aim is now to choose a matching~$M\in {\mathcal M}$ with $u_p(M)$ ``as close as possible'' to~$x_p$ for each country~$p$.
The difference $x_p-u_p(M)$ will then be taken as (positive or negative) {\it credits} for country~$p$ to the next round. In the next round we repeat the same steps after updating the underlying partitioned matching game $(N,v)$.

For setting the target allocations, we start with choosing, for the round under consideration, a ``fair'' {\it initial (kidney) allocation~$y$} that is prescribed by some solution concept for the underlying directed partitioned matching game $(N,v)$ for that round, so $y(N)=v(N)$ holds. Then, for each country~$p$, we change $y_p$ into $p$'s target $x_p=y_p+c_p$, where $c_p$ are the credits for country~$p$ from the previous round (we set $c\equiv 0$ for the first round). The initial allocation $y$ could correspond to a core allocation of the game $(N,v)$. We give more specific examples later.

We now define what we mean with being ``as close as possible'' to the target allocation~{$x$}. For country~$p$, we say that $|x_p-u_p(M)|$ is the {\it deviation} of country~$p$ from its target $x_p$ in a certain round. We may want to select a {\it minimal} matching from ${\cal M}$, which is a matching $M\in {\cal M}$ that minimizes
$$\max_{p\in N} \{|x_p-u_p(M)|\}$$
over all matchings of ${\cal M}$.
{We call such a matching $M$ {\it weakly close to  $x$}}.

As a more refined selection, we may want to select a matching from ${\cal M}$ that is even lexicographically minimal. That is, let
$$d(M)= (|x_{p_1}-u_{p_1}(M)|, \dots, |x_{p_n}-u_{p_n}(M)|)$$ be the vector obtained by reordering the components $|x_p-u_p(M)|$ non-increasingly.
We say that $M$ is {\it strongly close to}~$x$ if $d(M)$ is lexicographically minimal over all matchings $M\in {\cal M}$. Note that every {strongly close} matching in ${\cal M}$ is {weakly close}, but the reverse might not be true.

From now on, if $w\equiv 1$ on $G$, we consider the {\it actual} number of incoming kidneys that country~$p$ will receive in a certain round if $M$ is used. That is, we let for every $p\in N$, $u_p(M)$ be equal to $$s_p(M)=|\{j\in V_p |\; \exists i\in V: ij\in M\}|.$$
We say that the vector $(s_1(M),\ldots, s_n(M))$ is the {\it actual} allocation for a certain round. Note that $s_1(M)+\cdots + s_n(M)=2v(N)$ for every $M\in {\cal M}$, so formally we should consider $u_p=\frac{1}{2}s_p$ instead of~$s_p$. However, for our proofs it does not matter. Moreover, by defining $s_p$ as above we have a more natural definition in terms of actual numbers of kidneys. This is a natural utility function, both due to its simplicity and  because in practice the weights $w(e)$ are sparsely spread, as we explained above.

\medskip
\noindent
{\it Example 2.} Consider the directed compatibility graph ${\overline{G}}=(V,A)$ with the edge weighting~$w$ from Figure~\ref{f-stexample}. Let $V_1=\{i_1,i_2,i_3\}$ and $V_2=\{j_1,j_2\}$.
It holds that ${{\mathcal M}=\{M_1,M_2\}}$, where ${M_1=\{i_2j_2\}}$ and ${M_2=\{i_1i_2,j_1j_2\}}$. Note that ${v(N)=w(M_1)=w(M_2)=4}$. We have that  $u_1(M_1)=3$ and $u_2(M_1)=1$, whereas
$u_1(M_2)=u_2(M_2)=2$. We now set $w\equiv 1$ and consider the corresponding undirected compatibility graph $G=(V,E)$, which is also displayed in Figure~\ref{f-stexample}.
It now holds that  ${{\mathcal M}=\{M_2\}}$. Note that $v(N)=|M_2|=2$. We have $s_1(M_2)=s_2(M_2)=2$, so $s_1(M_2)+s_2(M_2)=2v(N)$, as we count ``incoming'' kidneys if $w\equiv 1$ on $G$. \dia

\medskip
\noindent
Recall that the aim of this paper is to provide a theoretical basis for our understanding of partitioned matching games. We refer to a number of recent papers~\cite{BBKP22,BGKPPV20,KNPV20} for experimental results on international kidney exchange that were obtained from simulations using uniform partitioned matching games. As initial allocations, Benedek et al.~\cite{BBKP22}
 used two ``hard-to-compute'' solution concepts, namely the Shapley value and nucleolus, and two easy-to-compute solution concepts, the benefit value and contribution value. They found that almost all initial and actual allocations in their simulations were core allocations, with the Shapley value yielding the best results. The latter finding was in line with the findings of~\cite{BGKPPV20} and~\cite{KNPV20}. In both \cite{BGKPPV20} and~\cite{KNPV20}, $3$-way exchanges were allowed at the expense of a  significant lower number of countries than the study of \cite{BBKP22}.

In line with our research aim we prove, apart from the results in Section~\ref{s-new}, a number of computational complexity results for computing {strongly and weakly close} maximum (weight) matchings. The first of these results, the main result of Section~\ref{s-app}, is a polynomial-time algorithm that was used for the simulations in~\cite{BBKP22}. For the proof of this result we refer to Section~\ref{s-results2}.

\begin{theorem}\label{t-easy1}
It is possible to find a {strongly close} maximum matching for a uniform partitioned matching game $(N,v)$ and target allocation $x$ in polynomial time.
\end{theorem}

\noindent
In the nonuniform case, we can show the following positive but weaker result than Theorem~\ref{t-easy1}. It only holds for directed partitioned matching games of width~$1$ {(i.e., for directed matching games)}.  Its proof can be found in Section~\ref{s-results2} as well.

\begin{theorem}\label{t-easy2}
It is possible to find a {strongly close} maximum weight matching for a directed partitioned matching game $(N,v)$ of width~$1$ {(i.e., for a directed matching game)} and target allocation $x$ in polynomial time.
\end{theorem}

\noindent
Our remaining results, all proven in Section~\ref{s-results3}, are all hardness results. They show that the situation quickly becomes computationally more difficult when weights are involved, even if we only want to find {a weakly close} maximum weight matching and the input is restricted in some way. For example, in the first of these hardness results we set the number of countries~$n$ equal to~$2$.

\begin{theorem}\label{t-hard1}
It is \NP-hard to find a {weakly close} maximum weight matching for a directed partitioned matching game~$(N,v)$ with $n=2$ and target allocation~$x$.
\end{theorem}

\noindent
Recall that the setting of sparsely spread edge weights is highly relevant in our context. Hence, we consider this situation more carefully.
We say that a directed partitioned matching game $(N,v)$, defined on a graph $G$ with positive edge weighting $w$ and partition ${\cal V}$ of $V$, is {\it $d$-sparse} if $w$ takes on only $d$ distinct values. Note that if $d=1$, we obtain a uniform game (possibly after rescaling~$w$) and in that case we can apply Theorem~\ref{t-easy1}. If $d=2$, then we already obtain computational hardness, as shown in our next result.

\begin{theorem}\label{t-hard2}
It is \NP-hard to find a {weakly close} maximum weight matching for a $2$-sparse directed partitioned matching game $(N,v)$ and target allocation~$x$.
\end{theorem}

\noindent
As will be clear from its proof, both the width~$c$ and the number of countries~$n$ in the hardness construction of Theorem~\ref{t-hard2}
{are} arbitrarily large. If $c$ and $n$ are constant, the partitioned matching game has constant size. Hence, it is natural to consider sparse directed partitioned matching games $(N,v)$ for which
either $c$ or $n$ is a constant. We were not able to solve either of these two cases, but we can show two partial results for the case where $n$ is a constant.

\medskip
\noindent
By a ``compact description'' of a game defined on a graph we mean a logarithmic description of the graph (if possible). For example, a cycle of length $k$ can be described by its length, which results in input size $O(\log k)$ rather than $k$. If we assume that we have such a compact description of a directed partitioned matching game, then having a constant number of countries
does not make the problem easier, as our next result shows.

\begin{theorem}\label{t-hard3}
It is \NP-hard to find a {weakly close} maximum weight matching for a $3$-sparse compact
directed partitioned matching game $(N,v)$
with $n=2$
and target allocation~$x$.
\end{theorem}

\noindent
We now make a connection to the {\sc Exact Perfect Matching} problem introduced by Papadimitriou and Yannakakis~\cite{PY82} forty years ago. This problem has as input an undirected graph $G$ whose edge set is partitioned into a set $R$ of {\it red} edges and a set $B$ of {\it blue} edges. The question is whether $G$ has a perfect matching with exactly $k$ red edges for some given integer~$k$. The complexity status of {\sc Exact Perfect Matching} is a longstanding open problem, and so far only partial results were shown (see, for example,~\cite{GKMT17}).

We consider directed partitioned matching games $(N,v)$ even with $n=2$. As before, we let $\overline{G}=(V,A)$ be the underlying directed compatibility graph with a positive arc weighting $w$ and with partition $(V_1,V_2)$ of $V$. For every $2$-cycle $iji$ with $i\in V_1$ and $j\in V_2$, we set $w_{ij}=\frac{1}{3}$ and $w_{ji}=\frac{2}{3}$.
For every other $2$-cycle $iji$ in $\overline{G}$ (so where $i,j\in V_1$ or $i,j\in V_2$) we set $w_{ij}=w_{ji}=\frac{1}{2}$.  Note that
$w\equiv 1$ in the corresponding undirected compatibility graph $G=(V,E)$ for the weighting~$w$ given by $w(ij)=w_{ij}+w_{ji}$ for every edge $ij\in E$.\footnote{The exact values of $w_{ij}$ and $w_{ji}=1-w_{ij}$ do not matter as long as they differ from $\frac{1}{2}$ on arcs between $V_1$ and $V_2$. The case where $w\equiv \frac{1}{2}$ in $\overline{G}$ corresponds to a partitioned matching game on an (undirected) uniform compatibility graph $G$, so Theorem~\ref{t-easy1} would hold.}
We also assume that $G$ has a perfect matching. Hence, as $w\equiv 1$ in $G$, the set ${\mathcal M}$ consists of all perfect matchings of $G$. We say that $(N,v)$ is {\it perfect}. By {construction}, the perfect partitioned matching games $(N,v)$ {that we just created} are $3$-sparse and have $n=2$.
We show the following result.

\begin{theorem}\label{t-hard4}
{\sc Exact Perfect Matching} and the problem of finding a {weakly close} maximum weight matching for a $3$-sparse perfect directed partitioned matching game with $n=2${,} and target allocation~$x$
are polynomially equivalent.
\end{theorem}

\noindent
In Section~\ref{s-con} we finish our paper with a discussion on future work.

\section{The Proofs of Theorems~\ref{t-main1}--\ref{t-struc2}}\label{s-b}

In this section we prove Theorems~\ref{t-main1}--\ref{t-struc2}.
A $b$-matching $M$ in a graph $G$ with capacity function $b$ {\it covers} a vertex $u$ if $M$ contains an edge incident to $u$, whereas $M$ {\it saturates} $u$ if $M$ contains $b(u)$ edges incident to $u$.
We identify $M$ with the subgraph of $G$ {\it induced} by $M$ (that is, the subgraph of $G$ consisting of all edges in $M$ and all vertices covered by $M$). We speak about \emph{(connected) components}  of~$M$. For instance, for $b=1$, every edge $e \in M$ is a component.

We start with Theorem~\ref{t-main1}, which we restate below.

\medskip
\noindent
{\bf Theorem~\ref{t-main1} (restated).}
{\it  \emph{P2} and \emph{P3} are co-\NP-hard for uniform $b$-matching games if $b\leq 3$.}

\begin{proof}
The proof is by reduction from a variant of the {\sc $3$-Regular Subgraph} problem. This problem is to decide if a given graph has a $3$-regular subgraph (a graph is {\it $3$-regular} if every vertex has degree~3). This problem is \NP-complete even for bipartite graphs~\cite{St97}. We define the {\sc Nearly $3$-Regular Subgraph} problem. This problem is to decide whether a (non-bipartite) graph $G$ has an {\it nearly $3$-regular subgraph}, that is, a subgraph in which all vertices have degree~$3$ except for one vertex that must be of degree~$2$.

The {\sc Nearly $3$-Regular Subgraph} problem is also \NP-complete. Namely, we can reduce from {\sc $3$-Regular Subgraph} restricted to bipartite graphs.
Given a bipartite graph $(U \cup V, E)$, we construct the non-bipartite graph~$G$ consisting of $|E|$ disjoint copies of $(U\cup V, E)$ where in the copy corresponding to $e \in E$ the edge $e$ is subdivided by a new vertex $v_e$ (that is, we replace $e$ by~$v_e$ and make $v_e$ adjacent to the two end-vertices of $e$).
Now, $(U \cup V, E)$ has a $3$-regular subgraph if and only if $G$ has a nearly $3$-regular subgraph. Indeed, if there is a $3$-regular subgraph in $(U \cup V, E)$ that contains the edge $e$, there will be a nearly $3$-regular subgraph in $G$ whose degree~$2$ vertex is~$v_e$. Conversely, if there is a nearly $3$-regular subgraph in $G$, it must contain a vertex $v_e$ for some~$e$; otherwise the subgraph would be bipartite, but an almost $3$-regular graph cannot be bipartite.

\medskip
\noindent
As mentioned, we reduce from  {\sc Nearly $3$-Regular Subgraph} for non-bipartite graphs. Given an instance $G=(V,E)$ of the latter, we construct a graph $\overline{G}$ with vertex capacities $b(i) \le 3$ and edge weights $w=1$  as follows (see also Figure~\ref{fig:weighted}).

\begin{itemize}
\item Give each $v\in V$ capacity $b(v)=3$.
\item Make each $v\in V$ adjacent to new vertices $a_{v,1}$, $a_{v,2}$, $a_{v,3}$ with $b(a_{v,1})=b(a_{v,2})=b(a_{v,3})=2$.
\item Make each $a_{v,j}$ part of a triangle with two new vertices $c_{v,j}$ and $d_{v,j}$ with $b(c_{v,j})=b(d_{v,j})=2$.
\item For each $v\in V$ add a new vertex $a_v$ with $b(a_v)=3$ and make $a_v$ adjacent to $a_{v,1}$, $a_{v,2}$, $a_{v,3}$.
\item For each $v\in V$ add a new vertex $c_v$ with $b(c_v)=3$ and make $c_v$ adjacent to $c_{v,1}$, $c_{v,2}$, $c_{v,3}$.
\item For each $v\in V$ add a new vertex $d_v$ with $b(d_v)=3$ and make $d_v$ adjacent to $d_{v,1}$, $d_{v,2}$, $d_{v,3}$.
\item Make a new vertex $r$ called the \emph{root node} $r$ with $b(r)=1$ adjacent to every $v \in V$.
\end{itemize}

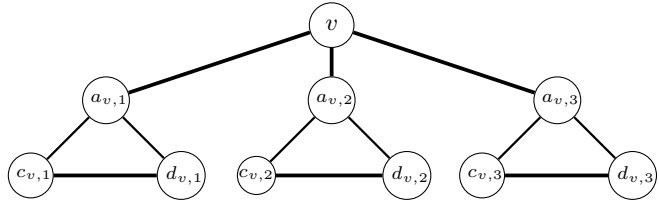
\begin{figure}
\begin{center}
  \begin{tikzpicture}[scale=1]
 \node[minimum size=1.5mm] (610)[draw=black,circle, inner sep=2pt] at (6,9){$~v~$};
 \scriptsize
 \node[minimum size=1.5mm] (38)[draw=black,circle, inner sep=0pt] at (3,8){$~a_{v,1}$};
 \node[minimum size=1.5mm] (68)[draw=black,circle, inner sep=0pt] at (6,8){$~a_{v,2}$};
   \node[minimum size=1.5mm] (98)[draw=black,circle, inner sep=0pt] at (9,8){$~a_{v,3}$};
 \node[minimum size=1.5mm] (27)[draw=black,circle, inner sep=0pt] at (2,7){$~c_{v,1}$};
 \node[minimum size=1.5mm] (47)[draw=black,circle, inner sep=0pt] at (4,7){$~d_{v,1}$};
\node[minimum size=1.5mm] (57)[draw=black,circle, inner sep=0pt] at (5,7){$c_{v,2}$};
 \node[minimum size=1.5mm] (77)[draw=black,circle, inner sep=0pt] at (7,7){$~d_{v,2}$};
 \node[minimum size=1.5mm] (87)[draw=black,circle, inner sep=0pt] at (8,7){$~c_{v,3}$};
   \node[minimum size=1.5mm] (107)[draw=black,circle, inner sep=0pt] at (10,7){$~d_{v,3}$};
\draw[line width=0.5mm]  (610)--(68);
\draw[line width=0.5mm]  (610)--(38);
\draw[line width=0.5mm]  (610)--(98);
\draw[line width=0.3mm]  (27)--(38)--(47);
\draw[line width=0.3mm]  (57)--(68)--(77);
\draw[line width=0.3mm]  (87)--(98)--(107);
\draw[line width=0.5mm] (27)--(47);
\draw[line width=0.5mm] (57)--(77);
\draw[line width=0.5mm] (87)--(107);
\end{tikzpicture}
\caption{A vertex $v\in V$ with the three pendant triangles ($a_v,c_v,d_v$). Thick edges are edges in $M$.
{In this case, we also say that $v$ is matched ``down'' to all of $a_{v,1},a_{v,2},a_{v,3}$.}
}\label{fig:weighted}
\end{center}
\end{figure}

We next describe a maximum $b$-matching  in $\overline{G}$, as indicated, in part, in Figure \ref{fig:weighted} as well.
Let~$M$ consist of all edges $va_{v,j}$ plus all edges of the form $c_{v,j}d_{v,j}$ plus all edges incident to $a_v,c_v$ and $d_v$.
As $M$ saturates all vertices except $r$, we find that $M$ is a maximum $b$-matching. Hence,
$$v(\overline{N})=|M|=3|V|+3|V|+9|V|=15|V|.$$
We let $(\overline{N},v)$ be the $b$-matching game defined on $(\overline{G},w)$. To show the theorem it remains to prove the following claim.

\medskip
\noindent
{\it Claim: $G$ has a nearly $3$-regular subgraph if and only if $(\overline{N},v)$ has an empty core.}

\medskip
\noindent
First suppose $G$ contains no nearly $3$-regular subgraph. We define a vector $x$ as follows.
\begin{itemize}
\item Set $x \equiv \frac{3}{2}$ on the vertices of $V$.
\item Set $x \equiv  1$ on all triangle vertices $a_{v,j}$, $c_{v,j}$, $d_{v,j}$.
\item Set $x \equiv \frac{3}{2}$ on all vertices $a_v,c_v,d_v$.
\item Set $x_r=0$.
\end{itemize}

\noindent
Note that $x(\overline{N})=\frac{3}{2}|V|+9|V|+3\cdot \frac{3}{2}|V|=15|V|=v(\overline{N})$. Hence, $x$ is an allocation.

We claim that $x$ is even a core allocation. For a contradiction, suppose $x(S)<v(S)$ for some $S \subseteq \overline{N}$; we say that $S$ is a {\it blocking} coalition. Let $M_S$ be a maximum $b$-matching in $\overline{G}[S]$, so  $x(S) < |M_S|$.
We assume that $S$ is a minimal blocking coalition (with respect to set inclusion).

As $x_i$ equals half the capacity of each vertex $i$ except $r$, we find that for every $S^*\subseteq N\setminus \{r\}$, we have $x(S^*)\geq v(S^*)$. Hence, $S$ contains $r$.
By the same reason, $M_S$ must saturate all vertices of $S$ (as otherwise $x(S)\geq |M_S|=v(S)$ would hold; {recall that $b(r)=1$}).
So, in particular $M_S$ contains some edge $rv_0$ for some $v_0 \in V$. As $M_S$ saturates all vertices of $S$, $v_0$ must be matched by $M_S$ to two more vertices (other than $r$).

Assume first that all $v \in S \cap V \setminus \{v_0\}$ are either matched ``down'' to $a_{v,1},a_{v,2},a_{v,3}$ by three matching edges in $M_S$
{(where ``down'' is with respect to the drawing in Figure~\ref{fig:weighted})}
 or matched ``up'' by three matching edges belonging to $E$.
Let $v \in V \cap S$ be matched {``}down{''} to its three triangles. Let
$A$ be the component of {$M_S$}
containing $v$. Note that $V(A)$ is a subset of the union of $\{v,a_v,c_v,d_v\}$ and the set of vertices of the three triangles for $v$.
Hence, as all vertices in $A$ are saturated by $M_S$ and each vertex except $r$ gets exactly half of its capacity, we find that $x(A)\geq v(V(A))$. So, $S\setminus V(A)$
is a smaller blocking set, contradicting the minimality of~$S$. Thus, all vertices $v \in S\cap V$ with $v \neq v_0$ must be matched ``up''. If also $v_0$ is matched ``up'' by two edges in $M_S \cap E$, then $(S\cap V,M_S)$ is a nearly 3-regular subgraph of $G$, a contradiction.

From the above, we are left to deal with the case where there exists a vertex $v \in S\cap V$ that is, say, matched {``}down{''} by some edge $e=va_{v,1}\in M_S$ but, say, $e'=va_{v,3}\notin M_S$. We distinguish the following three cases:

\medskip
\noindent
{\bf Case 1.} $a_v,c_v,d_v \in S$.\\
Since all these are saturated by $M_S$, we have all $a_{v,j}, c_{v,j}, d_{v,j}\in S$. Thus $a_{v,3},c_{v,3},d_{v,3} \in S$ and each of these is already
 matched to $a_v,c_v,d_v$, respectively. Since $va_{v,3} \notin M_S$, at most two of $a_{v,3},c_{v,3},d_{v,3}$ can be saturated by $M_S$,
 a contradiction.

\medskip
\noindent
{\bf Case 2.} $a_v,c_v \in S, d_v \notin S$.\\
Again we find that $a_{v,1},a_{v,2},a_{v,3} \in S$ and $c_{v,1},c_{v,2},c_{v,3} \in S$. Moreover,
 each of these is already matched by some edge in $M_S$ to $a_v$ or $c_v$. In addition, $a_{v,1}$ is matched to $v$, so $a_{v,1}$ is
 ``already'' saturated. Hence, in order to saturate also $c_{v,1}$, $M_S$ must contain $c_{v,1}d_{v,1}$. Hence, $d_{v,1} \in S$ and $M_S$
 cannot saturate it (as $d_v \notin S$), a contradiction.

\medskip
\noindent
 {\bf Case 3.}
 $a_v \in S, c_v, d_v \notin S$.\\
 Here, we conclude that $a_{v,1},a_{v,2},a_{v,3} \in S$. Since $va_{v,3} \notin M_S$, $a_{v,3}$ can  only
 be saturated if, say, $a_{v,3}c_{v,3} \in M_S$ and hence $c_{v,3} \in S$. The latter can only be saturated by $c_{v,3}d_{v,3} \in M_S$.
 Hence, $d_{v,3} \in S$ and $d_{v,3}$ cannot be saturated (since $a_{v,3}a_v$ and $a_{v,3}c_{v,3}$ are in $M_S$), a contradiction.

 \medskip
 \noindent
{\bf Case 4.} $a_v \notin S$.\\
Since $a_{v,1}$ is in $S$, it must be saturated, and as  $a_v\notin S$, either $a_{v,1}c_{v,1}$ or $a_{v,1}d_{v,1}\in M_S$. By symmetry, suppose that $a_{v,1}c_{v,1}\in M_S$. Then $c_{v,1}$ is in $S$ and must be saturated, so either $c_{v,1}c_v\in M_S$ or $c_{v,1}d_{v,1}\in M_S$. In the first case $c_v\in S$ and $d_v\notin S$ (as $d_{v,1}\notin S$) and in the second case $c_v\notin S$ and $d_v\in S$ (as $d_{v,1}$ can only be saturated by $d_{v,1}d_v$). In both cases we get a contradiction when considering the third triangle, as follows. If $c_v\in S$ and $d_v\notin S$, then we have $c_vc_{v,3}\in M_S$. Thus $c_{v,3}$ is in $S$ and must be saturated, so must be matched to either $a_{v,3}$ or $d_{v,3}$. In the former case $a_{v,3}$ must be matched to $d_{v,3}$ and the latter remains unsaturated, a contradiction. In the second case, $d_{v,3}$ must be saturated by matching it to $a_{v,3}$ and then again, the latter remains unsaturated. The case $d_v \in S$ and $c_v \notin S$ is similar. From $d_v\in S$ we conclude that $d_vd_{v,3} \in M_S$. Thus $d_{v,3} \in S$ and hence, $d_{v,3}$ must be matched to either $a_{v,3}$ or $c_{v,3}$. In the first case, $a_{v,3}$ must be matched to $c_{v,3}$ (as $v$ and $a_v$ are not available) and $c_{v,3}$ remains unsaturated. In the second case, $c_{v,3}$ must be matched to $a_{v,3}$ and again, $a_{v,3}$ remains unsaturated, a contradiction.

\medskip
\noindent
Conversely, suppose that $G'$ {is} a nearly $3$-regular subgraph in $G$. We claim that $(\overline{N},v)$ has an empty core. For a contradiction, suppose that $(\overline{N},v)$ has a core allocation~$x$.
Fix any  $v \in V$ and let $S_{ac}:= \{a_v,c_v,a_{v,j},c_{v,j}~| ~j=1, 2,3\}$. As $S_{ac}$ allows a saturating matching $M_{ac}$ of size $|M_{ac}|=9$,
we find that $x(S_{ac}) \ge 9$. Similarly, $x(S_{cd}) \ge 9$ and $x(S_{ad}) \ge 9$ for $S_{cd}$ and $S_{ad}$ defined
analogously. Adding all three inequalities and dividing by $2$ yields $$x(S_v) \ge \frac{27}{2}\; \mbox{for}\; S_v:= S_{ac} \cup S_{cd} \cup S_{ad}.$$
We now recall the maximum $b$-matching $M$ that we displayed, in part, in Figure \ref{fig:weighted}.
As $M$ is maximum and $x$ is a core allocation, it holds that $x(A)=v(A)$ for every component $A$ of $M$, and moreover, $x_r=0$.
The set $S_v \cup \{v\}$ is covered exactly by two components of $M$.
Hence, {$x(S_v \cup \{v\}) = 6+9=15$}. As we also have that $x(S_v)\geq \frac{27}{2}$, this implies that $x_v \le \frac{3}{2}$. As $v$ was chosen arbitrarily, the inequality holds for every vertex of $V$.

Now, let $v_0$ be the unique vertex that has degree $2$ in $G'$, and recall that by definition all other vertices of $G'$ have degree~$3$ in $G'$. Consider the coalition $V(G')\cup \{r\}$.
The edge set $E(G') \cup \{rv_0\}$ matches each vertex
in $S$ up to its capacity, while $x$ assigns only half this value to each vertex in~$S$ and zero to $r$. Hence, $x(S) < v(S)$, contradicting to our assumption that $x$ is in the core. \qed
\end{proof}

\begin{figure}
\begin{minipage}{0.48\linewidth}
\centering
\begin{tikzpicture}[scale=1.3,rotate=0]
\draw
(0, 1) node[circle, black, scale=0.8,draw](a){\small$1$}
(1.5, 2) node[label=above:{\small[$2$]},circle, black, scale=0.8,draw](b){\small$2$}
(1.5, 0) node[circle, black, scale=0.8,draw](c){\small$3$}
(3, 3) node[circle, black, scale=0.8,draw](d){\small$4$}
(3, 1) node[label=right:{\small[$3$]},circle, black, scale=0.8,draw](e){\small$5$}
(3, -1) node[circle, black, scale=0.8,draw](f){\small$6$};
\draw[-] (a) -- node[above] {\small$2$} (b);
\draw[-] (b) -- node[above] {\small$2$} (d);
\draw[-] (d) -- node[right] {\small$3$} (e);
\draw[-] (b) -- node[above] {\small$1$} (e);
\draw[-] (a) -- node[below] {\small$1$} (c);
\draw[-] (c) -- node[above] {\small$3$} (e);
\draw[-] (c) -- node[below] {\small$2$} (f);
\draw[-] (e) -- node[right] {\small$2$} (f);
\end{tikzpicture}
\end{minipage}
\begin{minipage}{0.48\linewidth}
\centering
\begin{tikzpicture}[scale=1.3,rotate=0]
\draw
(0, 1) node[circle, black, scale=0.8,draw](a){\small$1$}
(0.5,2/3) node[fill,circle, black, scale=0.3,draw](p){}
(1,1/3) node[fill,circle, black, scale=0.3,draw](q){}
(0.5,1+1/3 ) node[fill,circle, black, scale=0.3,draw](g){}
(1,1+2/3 ) node[fill,circle, black, scale=0.3,draw](h){}
(1, 2.45) node(b){\small$V_{2}$}
(1.75,2-1/6) node[circle, black, scale=0.6,draw](x){}
(1.25,2+1/6) node[circle, black, scale=0.6,draw](y){}
(2,2-1/3) node[fill,circle, black, scale=0.3,draw](v){}
(2.5,2-2/3) node[fill,circle, black, scale=0.3,draw](w){}
(2,2+1/3) node[fill,circle, black, scale=0.3,draw](i){}
(2.5,2+2/3) node[fill,circle, black, scale=0.3,draw](j){}
(1.5, 0) node[circle, black, scale=0.8,draw](c){\small$3$}
(2,-1/3) node[fill,circle, black, scale=0.3,draw](r){}
(2.5,-2/3) node[fill,circle, black, scale=0.3,draw](s){}
(2,1/3) node[fill,circle, black, scale=0.3,draw](t){}
(2.5,2/3) node[fill,circle, black, scale=0.3,draw](u){}
(3, 3) node[circle, black, scale=0.8,draw](d){\small$4$}
(3,1+4/3) node[fill,circle, black, scale=0.3,draw](l){}
(3,1+2/3) node[fill,circle, black, scale=0.3,draw](m){}
(3.8, 1) node(e){\small$V_{5}$}
(3.25,1+1/6) node[circle, black, scale=0.6,draw](z){}
(3.25,5/6) node[circle, black, scale=0.6,draw](aa){}
(2.75,1) node[circle, black, scale=0.6,draw](bb){}
(3,-1+4/3) node[fill,circle, black, scale=0.3,draw](n){}
(3,-1+2/3) node[fill,circle, black, scale=0.3,draw](o){}
(3, -1) node[circle, black, scale=0.8,draw](f){\small$6$};
\draw[-] (y) to[out=0,in=90] (v);
\draw[-] (a) -- (g)-- (h)-- (y)-- (i)-- (j)-- (d)-- (l)-- (m)-- (bb)-- (n)-- (o)-- (f)-- (s)-- (r)-- (c)-- (q)-- (p)-- (a);
\draw[-] (c) -- (t)-- (u)-- (bb);
\draw[-] (x) -- (v)-- (w)-- (bb);
\draw[-] (h) -- (x)-- (i);
\draw[-] (m) -- (z);
\draw[-] (m) to[out=-30,in=30] (aa);
\draw[-] (w) -- (z);
\draw[-] (w) to[out=-25,in=145] (aa);
\draw[-] (u) to[out=25,in=-145] (z);
\draw[-] (u) -- (aa);
\draw[-] (n) to[out=30,in=-30] (z);
\draw[-] (n) -- (aa);
\draw[dashed] (3.1,1) ellipse (0.5 and 0.35);
\draw[dashed] (1.5,2) ellipse (0.5 and 0.38);
\draw[densely dotted,thick] (0,1) ellipse (0.4 and 0.4) node[below,yshift=-.4cm, xshift=-.2cm]{\scriptsize $V_{1}$};
\draw[densely dotted,rotate around={145:(2.25,-0.5)},thick] (2.25,-0.5) ellipse (0.5 and 0.35) node[ left,xshift=-.2cm,yshift=-0.4cm]{\scriptsize $\{3_6,6_3\}$};
\end{tikzpicture}
\end{minipage}
\caption{{\it Left}: a $b$-matching game $(N,v)$ with six players, where $b\equiv 1$ apart from
$b(2)=2$ and $b(5)=3$, so $b^*=3$.
Note that $v(N)=10$ (take $M=\{12,35,45,56\}$).
{\it Right:} the reduction to the partitioned matching game~$(\overline{N},\overline{v})$. Note that
$|\overline{N}|=14$ and $c=b^*$ (example taken from~\cite{BBJPX}.)}\label{f-second}
\end{figure}
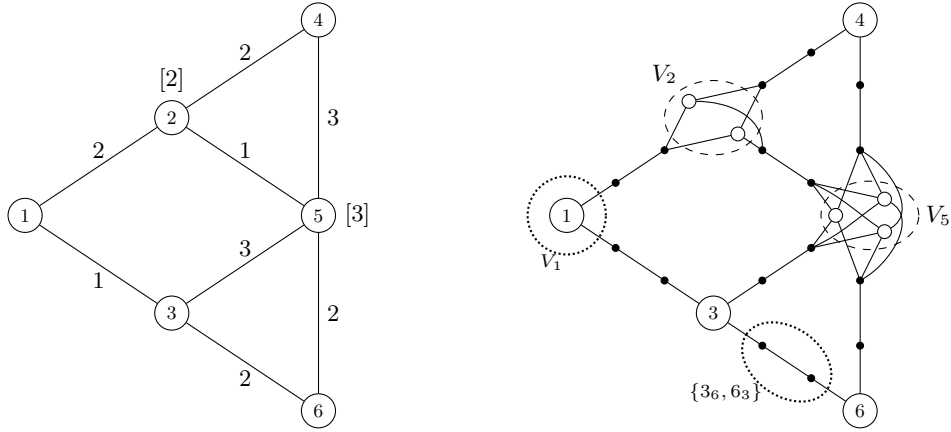

\noindent
We will now prove Theorem~\ref{t-struc1}. In order to do this we first explain our reduction. Let $(N,v)$ be a $b$-matching game defined on a graph $G=(V,E)$ with a positive capacity function~$b$ and a positive edge weighting~$w$.
{We assume that $b^* \ge 2$, as otherwise we obtain a matching game.}
We construct a graph $\overline{G}=(\overline{V},\overline{E})$ with a positive edge weighting~$\overline{w}$ and partition~${\cal V}$ of $\overline{V}$ by applying the construction of Tutte \cite{Tu54}:

\begin{itemize}
\item Replace each vertex $i \in V$ with capacity $b(i)$ by a set $V_i$ of $b(i)$ vertices $i^1,\ldots,i^{b(i)}$.
\item Replace  each edge $ij \in E$ by an edge $i_jj_i$ where $i_j$ and $j_i$ are new vertices, such that $i_j$ is adjacent to $i^1,\ldots,i^{b(i)}$, while $j_i$ is adjacent to  $j^1,\ldots,j^{b(j)}$. Let $E_{ij}=\{i_j,j_i\}$.
\item Let $\overline{V}$ consist of all the vertices in the sets $V_i$ and the sets $E_{ij}$, and let $\overline{E}$ consist of all the edges defined above.
\item Give every edge $i_jj_i$, $i_ji^h$ $(h\in \{1,\ldots,b(i)\})$, $j_ij^h$ $(h\in \{1,\ldots,b(j)\})$ weight $w(ij)$ to obtain the weighting $\overline{w}$.
\item  Let ${\cal V}$ consist of all the sets $V_i$ and $E_{ij}$.
\end{itemize}

\noindent
We denote the partitioned matching game defined on $(\overline{G},{\overline{w}})$ with partition~${\cal V}$ by $(\overline{N},\overline{v})$. See Figure~\ref{f-second} for an example.
We let $b^*$ be the maximum $b(i)$-value for the vertex capacity function~$b$.
Note that, by construction {and our assumption that $b^*\ge 2$}, we have that $c=b^*$ and that uniformity is preserved.
We use the above construction to prove Theorem~\ref{t-struc1}, which we restate below.

\medskip
\noindent
{\bf Theorem~\ref{t-struc1} (restated).}
{\it \emph{P1}--\emph{P3} can be reduced in polynomial time from $b$-matching games to partitioned matching games of width $c=b^*$, preserving uniformity.}

\begin{proof}
{Recall that we assume $b^* \ge 2$, as for $b^*=1$ both problems are identical.}
 Let $(N,v)$ be a $b$-matching game defined on a graph $G=(V,E)$ with a positive capacity function~$b$ and a positive edge weighting~$w$. We construct in polynomial time the pair $(\overline{G},\overline{w})$ and partition~${\cal V}$ to obtain the partitioned matching game $(\overline{N},\overline{v})$. Recall that the players of $\overline{N}$ are of the form $V_i$ and $E_{ij}$, and they are
 in $1-1$ correspondence with $V=N$ and $E$, respectively, so below we sometimes will identify the players of $\overline{N}$ with $V \cup E$.

The idea is that any $b$-matching $M$ in $G$ can be {\it represented} by a corresponding matching $\overline{M}\subseteq \overline{E}$ in $\overline{G}$ as follows.
For each edge $ij$ in $G$ we do as follows.
If $ij\in M$, then we match
$i_j$ to some copy of $i$ in $\overline{G}$ and, similarly, $j_i$ to some copy of $j$; note that, by definition, enough copies of $i$ and $j$ are
available. If  $ij\notin M$, then we match $i_j$ and $j_i$ to each other
in $\overline{G}$. We refer to the resulting matching $\overline{M}$  in $\overline{G}$ as a \emph{transform} of $M$ (different transforms differ by the choice of copies of vertex $i$ that are ``matched'' to $j$).
Note that $\overline{M}$ has size $|E|+|M|$ and weight $w(E)+w(M)$.

\medskip
\noindent
We first reduce P1. Let $x$ be an allocation of $(N,v)$. We define a vector $\overline{x}$ on $\overline{N}$ by setting $\overline{x}(V_i) := x_i$ for every $i \in N$ and $\overline{x}(E_{ij}) = w(ij)$ for every $ij\in E$. We claim that $x \in \mathit{core}(N,v)$ if and only if $\overline{x}\in \mathit{core}(\overline{N},\overline{v})$.

First suppose that $x \in \mathit{core}(N,v)$. Assume that $M$ is a maximum weight $b$-matching in $G$ (so $v(N)=w(M)$).
The transform~$\overline{M}$ of $M$ is a maximum weight matching in $\overline{G}$, so $$\overline{v}(\overline{N})=
\overline{w}(\overline{M})= w(E)+w(M)=\overline{x}(\overline{N}).$$
Thus $\overline{x}$ is an allocation.
{To prove the core constraints, suppose for a contradiction that there exists a blocking coalition $\overline{S^*} \subseteq \overline{N}$, i.e.,  $\overline{v}(\overline{S^*})>\overline{x}(\overline{S^*})$. For a coalition $S\subseteq N$, let the \textit{projection} of $S$ be defined as $\overline{S}=\bigcup\{V_i\; |\; i\in S\} \cup \bigcup \{E_{ij}\; |\; i,j \in S\}$ in $\overline{N}$. First we show that if $\overline{S^*}$ is not a projection of any coalition $S\subseteq N$ then we can modify $\overline{S^*}$ by removing and adding players of form $E_{ij}$ to obtain another blocking coalition $\overline{S'}\subseteq \overline{N}$ that is a projection of a coalition $S'\subseteq N$.}

{There are two cases to consider. First, if $V_i,V_j\in \overline{S^*}$, but $E_{ij}\notin \overline{S^*}$, then for $\overline{S'}=\overline{S^*}\cup E_{ij}$ we get $\overline{v}(\overline{S'})\geq \overline{v}(\overline{S^*})+w(ij)$, whilst $\overline{x}(\overline{S'})= \overline{x}(\overline{S^*})+w(ij)$, so $\overline{S'}$ is also blocking. In the second case, if $E_{ij}\in \overline{S^*}$, but one of $V_i$, $V_j$, say $V_{i}\notin \overline{S^*}$, then for $\overline{S'}=\overline{S^*}\setminus E_{ij}$ we have $\overline{v}(\overline{S'})= \overline{v}(\overline{S^*})-w(ij)$ and $\overline{x}(\overline{S'})= \overline{x}(\overline{S^*})-w(ij)$, so $\overline{S'}$ is also blocking. Therefore, if a coalition $\overline{S^*}$ is blocking, then it will remain blocking when making the above operations until we obtain a blocking coalition $\overline{S}$ that is a projection of a coalition $S\subseteq N$.}

{Finally, consider a blocking coalition $\overline{S} \subseteq \overline{N}$ that is a projection of $S\subseteq N$, and let $V_S$ and $\overline{V}_{\overline{S}}$ denote the corresponding vertex sets in $G$ and $\overline{G}$, respectively. By definition, $v(S)$ is the maximum weight of a $b$-matching $M$ in subgraph $G_S=G[V_S]$ and $\overline{v}(\overline{S})$ is the maximum weight of a matching $\overline{M}$ in $\overline{G}[\overline{V}_{\overline{S}}]$, where $\overline{M}$ is a transform of $M$. Hence, $$w(E(G_S))+x(V_S)=\overline{x}(\overline{S})< \overline{v}(\overline{S})= w(E(G_S))+w(M)=w(E(G_S))+v(S),$$ which implies $x(S)<v(S)$, a contradiction.}

Now suppose that $\overline{x}\in \mathit{core}(\overline{N},\overline{v})$. Then we can use the same arguments in reverse order to prove that $x \in \mathit{core}(N,v)$.
{Finally, if $\overline{x}\notin \mathit{core}(\overline{N},\overline{v})$ and one can find a blocking coalition $\overline{S}$ in $\overline{N}$, then by the above argument we can also find a blocking coalition $S$ for $x$ in $(N,v)$, completing the proof for P1.}

{The reduction of P3 is implied directly by the above argument for P1. The transformation of a core allocation $\overline{x}$ for $(\overline{N},\overline{v})$ into a core allocation $x$ for $(N,v)$ takes polynomial time. Hence, if we can find a core allocation for $(\overline{N},\overline{v})$ in polynomial time, then we can also find a core allocation for $(N,v)$ in polynomial time.}

\medskip
\noindent
To reduce P2, it suffices to show that $\mathit{core}(N,v)$ is non-empty if and only if $\mathit{core}(\overline{N},\overline{v})$ is non-empty. Above we already showed that if $\mathit{core}(N,v)$ is non-empty, then $\mathit{core}(\overline{N},\overline{v})$ is non-empty.
Hence, assume $\mathit{core}(N,v)= \emptyset$.
By the Bondareva-Shapley Theorem~\cite{Bo63,Sh67}, {$\mathit{core}(N,v)= \emptyset$ if and only if} there is a function $\lambda: 2^N\setminus \{\emptyset\}\ \to [0,1]$ with for every $i\in N$, $\sum_{S\ni i} \lambda(S)=1$ such that
$\sum_{S\neq \emptyset}\lambda(S)v(S)>v(N)$.
For every non-empty coalition $S\subseteq N$, we define the set
$$\overline{S}:= \bigcup \{V_i ~| ~i \in S\} \cup \bigcup \{E_{ij} ~| ~i,j \in S\}.$$
We now define a function $\overline{\lambda}:2^{\overline{N}}\setminus \{\emptyset\} \to [0,1]$ as follows.
We set $\overline{\lambda}(\overline{S}):=\lambda(S)$ for every $S\subseteq N$. For every $E_{ij}\in \overline{N}$, we set $$\overline{\lambda}(E_{ij}):= 1- \sum_{\overline{S}: E_{ij} \subseteq \overline{S}} \overline{\lambda}(\overline{S})=1-\sum_{S:ij\in E(S)}\lambda(S).$$
Finally, set $\overline{\lambda}(S^*):=0$ for all other $S^*\subseteq \overline{N}\setminus \{\emptyset\}$.
By construction, we have $\sum_{S^*\subseteq \overline{N}, V_i\in S^*}\overline{\lambda}(S^*)=1$ for every $V_i\in \overline{N}$ and $\sum_{S^*\subseteq \overline{N}, E_{ij}\in S^*}\overline{\lambda}(S^*)=1$ for every $E_{ij}\in \overline{N}$.
Now, in order to show that $\mathit{core}(\overline{N},\overline{v})$ is empty, we must prove that $$\sum_{\emptyset\neq S^*\subseteq \overline{N}}\overline{\lambda}(S^*)\overline{v}(S^*)>\overline{v}(\overline{N}).$$
Let $M_S$ denote a maximum weight $b$-matching on $G[S]$, and let $\overline{M_S}$ be a transform of $M_S$, which is a maximum weight matching in $\overline{G}[\overline{S}]$ with weight $\overline{w}(\overline{M_S})=w(M_S)+w(E(G[S]))$. So, we have $$\overline{v}(\overline{S})=\overline{w}(\overline{M_S})=w(M_S)+w(E(G[S]))=v(S)+w(E(G[S])).$$ Hence, it follows that
$$\begin{array}{lcl}
&&\sum_{\emptyset\neq S^*\subseteq \overline{N}}\overline{\lambda}(S^*)\overline{v}(S^*)\\[8pt] &= &\sum_{\emptyset\neq S\subseteq N}\overline{\lambda}(\overline{S})\overline{v}(\overline{S})+ \sum_{ij\in E(G)}\overline{\lambda}(E_{ij})\overline{v}(E_{ij})\\[8pt]
&= &\sum_{\emptyset\neq S\subseteq N}\lambda(S)(v(S)+w(E[S]))+\sum_{ij\in E(G)}\overline{\lambda}(E_{ij})w(ij)\\[8pt]
&= &\sum_{\emptyset\neq S\subseteq N}\lambda(S)v(S)+ \sum_{ij\in E(G)}\sum_{S:ij\in E(S)}\lambda(S)w(ij)+\sum_{ij\in E(G)}\overline{\lambda}(E_{ij})w(ij)\\[8pt]
&= &\sum_{\emptyset\neq S\subseteq N}\lambda(S)v(S)+ \sum_{ij\in E(G)}\left[\sum_{S:ij\in E(S)}\lambda(S)+\overline{\lambda}(E_{ij})\right]w(ij)\\[8pt]
& =&\sum_{\emptyset\neq S\subseteq N}\lambda(S)v(S)+w(E)\\[8pt]
& > &v(N)+ w(E)\\[8pt]
&=&\overline{v}(\overline{N}).
\end{array}$$
{By the Bondareva-Shapley Theorem, this implies that $\mathit{core}(\overline{N},\overline{v})= \emptyset$.}
This completes the proof of Theorem~\ref{t-struc1}. \qed
\end{proof}

\noindent
As the final result in this section, we will prove Theorem~\ref{t-struc2}. Again, we first explain the reduction.
Let $(N,v)$ be a partitioned matching game of width~$c$ defined on a graph $G=(V,E)$ with a positive edge weighting $w$ and with a partition
$(V_1, \dots,V_n)$ of $V$. We assume that $c \ge 2$, as otherwise we obtain a matching game. Construct a graph $\overline{G}=(\overline{N},\overline{E})$  with a positive vertex capacity function~$b$ and a positive edge weighting~$\overline{w}$
as follows.

\begin{itemize}
\item Put the vertices of $V$ into ${\overline N}$ and the edges of $E$ into $\overline{E}$.
\item For each $V_i$, add a new vertex $r_i$ to $\overline{N}$ that is adjacent to all vertices of $V_i$ and to no other vertices in $\overline{G}$.
\item Let $\overline{w}$ be the extension of $w$ to $\overline{E}$ by giving each new edge
weight $2v(N)$.
\item In $\overline{V}$, give every $u\in V$ capacity $b(u)=2$ and every $r_i$ capacity $b(r_i)=|V_i|$.
\end{itemize}

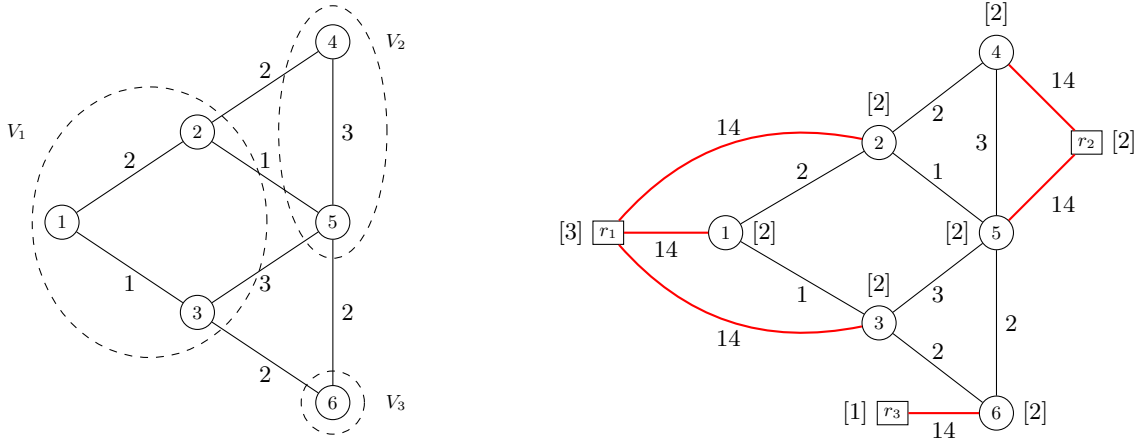
\begin{figure}
\begin{minipage}{0.48\linewidth}
\begin{tikzpicture}[scale=1.2,rotate=0]
\draw
(0, 1) node[circle, black, scale=0.8,draw](a){\small$1$}
(1.5, 2) node[circle, black,scale=0.8, draw](b){\small$2$}
(1.5, 0) node[circle, black, scale=0.8,draw](c){\small$3$}
(3, 3) node[circle, black, scale=0.8,draw](d){\small$4$}
(3, 1) node[circle, black, scale=0.8,draw](e){\small$5$}
(3, -1) node[circle, black, scale=0.8,draw](f){\small$6$}
(-0.5, 2) node[scale=0.8](h){\small$V_{1}$}
(3.7, 3) node[scale=0.8](i){\small$V_{2}$}
(3.7, -1) node[scale=0.8](j){\small$V_{3}$};
\draw[-] (a) -- node[above] {\small$2$} (b);
\draw[-] (b) -- node[above] {\small$2$} (d);
\draw[-] (d) -- node[right] {\small$3$} (e);
\draw[-] (b) -- node[above] {\small$1$} (e);
\draw[-] (a) -- node[below] {\small$1$} (c);
\draw[-] (c) -- node[below] {\small$3$} (e);
\draw[-] (c) -- node[below] {\small$2$} (f);
\draw[-] (e) -- node[right] {\small$2$} (f);
\draw[dashed] (0.97,1) ellipse (1.3 and 1.5);
\draw[dashed] (3,2) ellipse (0.6 and 1.4);
\draw[dashed] (3,-1) ellipse (0.35 and 0.35);
\end{tikzpicture}
\end{minipage}
\begin{minipage}{0.48\linewidth}
\centering
\begin{tikzpicture}[scale=1.2,rotate=0]
\draw
(-1.3, 1) node[label=left:{\small[$3$]},rectangle,black, scale=0.8,draw](g){\small$r_{1}$}
(0, 1) node[label=right:{\small[$2$]},circle, black, scale=0.8,draw](a){\small$1$}
(1.7, 2) node[label=above:{\small[$2$]},circle, black,scale=0.8,draw](b){\small$2$}
(1.7, 0) node[label=above:{\small[$2$]},circle, black, scale=0.8,draw](c){\small$3$}
(3, 3) node[label=above:{\small[$2$]},circle, black, scale=0.8,draw](d){\small$4$}
(3, 1) node[label=left:{\small[$2$]},circle, black, scale=0.8,draw](e){\small$5$}
(4, 2) node[label=right:{\small[$2$]},rectangle,black, scale=0.8,draw](h){\small$r_{2}$}
(1.85, -1) node[label=left:{\small[$1$]},rectangle,black, scale=0.8,draw](i){\small$r_{3}$}
(3, -1) node[label=right:{\small[$2$]},circle, black, scale=0.8,draw](f){\small$6$};
\draw[-] (a) -- node[above] {\small $2$} (b);
\draw[-] (b) -- node[below] {\small$2$} (d);
\draw[-] (d) -- node[left] {\small$3$} (e);
\draw[-] (b) -- node[above] {\small$1$} (e);
\draw[-] (a) -- node[below] {\small$1$} (c);
\draw[-] (c) -- node[below] {\small$3$} (e);
\draw[-] (c) -- node[above] {\small$2$} (f);
\draw[-] (e) -- node[right] {\small$2$} (f);
\path[red][thick] (g) [bend left]edge node[above,black]{\small$14$} (b);
\draw[red][thick] (g) -- node[below,black] {\small$14$} (a);
\path[red][thick] (g) [bend right] edge node[below,black] {\small$14$} (c);
\draw[red][thick] (d) -- node[above right,black] {\small$14$} (h);
\draw[red][thick] (e) -- node[below right,black] {\small$14$} (h);
\draw[red][thick] (i) -- node[below,black] {\small$14$} (f);
\end{tikzpicture}
\end{minipage}
\caption{{\it Left:} a partitioned matching game $(N,v)$ with three players and width $c=3$. Note that $v(N)=7$.
{\it Right:} the reduction to the $b$-matching game $(\overline{N},\overline{v})$. Note that $|\overline{N}|=9$ and that for every $i\in \overline{N}$, $b(i)\leq c$ (example taken from~\cite{BBJPX}).}\label{f-third}
\end{figure}

\noindent
We denote the $b$-matching game defined on $(\overline{G},b,\overline{w})$ by $(\overline{N},\overline{v})$. See Figure~\ref{f-third} for an example and note that, by construction, we have $b\leq c$.
We are now ready to prove Theorem~\ref{t-struc2}, which we restate below.

\medskip
\noindent
{\bf Theorem~\ref{t-struc2} (restated).}
{\it  \emph{P1}--\emph{P3} can be reduced in polynomial time from partitioned matching games of width~$c$ to $b$-matching games for some capacity function $b$ with $b^*\leq c$.}

\begin{proof}
Recall that we assume $c \ge 2$, as for $c=1$ both problems are identical.
 Let $(N,v)$ be a generalized matching game defined by a graph $G=(V,E)$ with edge weights $w$ and partition $\mathcal{V}=(V_1, \dots,V_n)$.
 We construct in polynomial time the triple $(\overline{G},b,\overline{w})$  to obtain the $b$-matching game $(\overline{N},\overline{v})$.

\medskip
\noindent
We first reduce P1. Let $x$ be an allocation of $(N,v)$. We define a vector $\overline{x}$ on $\overline{N}$ by setting for
every $i\in N$, $\overline{x} :\equiv \frac{x_i}{|V_i|} +v(N)$ on $V_i$ and $\overline{x}_{r_i} := v(N)\cdot |V_i|$.
We claim that $x \in \mathit{core}(N,v)$ if and only if $\overline{x}\in \mathit{core}(\overline{N},\overline{v})$.

First suppose that $x \in \mathit{core}(N,v)$. By construction, $\overline{x}(\overline{N})=v(N)+2v(N)\cdot|V|=\overline{v}(\overline{N})$.
 To check the core constraints, consider a coalition $\overline{S} \subseteq \overline{N}$. Let $S := \{i~|~\overline{S}\cap V_i \neq \emptyset\}$. A maximum weight $b$-matching in $\overline{G}[\overline{S}]$ is obtained by matching each root node $r_i \in \overline{S}$ to all its neighbors in $\overline{S}$ and matching
 the nodes in $\overline{S} \cap V$ to each other in the best possible way. Thus
 $$\overline{v}(\overline{S})\le v(S) + \displaystyle\sum_{i: r_i \in \overline{S}} 2v(N)|\overline{S} \cap V_i|,$$
while
$$\displaystyle\overline{x}(\overline{S}) = {\sum_{i \in S} |\overline{S} \cap V_i| \left( \frac{x_i}{|V_i|} + v(N)\right)} + \sum_{i: r_i \in \overline{S}} v(N)|V_i|.$$
{For a coalition $S\subseteq N$, let $\overline{S}= \bigcup_{i \in S} V_i \cup \{r_i\}$ denote the projection of $S$ in $\overline{N}$.
We show that if the core constraint $\overline{x}(\overline{S^*}) \ge \overline{v}(\overline{S^*})$ is violated for a coalition $\overline{S^*}\subset \overline{N}$, then it would also be violated for a coalition $\overline{S}\subseteq \overline{N}$, where $\overline{S}$ is the projection of $S\subseteq N$. There are two cases to consider. First, suppose that for some $j\in V_i$, $j\in \overline{S^*}$ but $r_i\notin \overline{S^*}$. Then, for $\overline{S'}=\overline{S^*}\setminus \{j\}$, we have $x(\overline{S'})< x(\overline{S^*})-v(N)$, whilst $v(\overline{S'})\geq v(\overline{S^*})-v(N)$. So, removing $j$ from $\overline{S^*}$ would preserve the violation. In the second case, suppose that $r_i\in \overline{S^*}$, but $j\notin \overline{S^*}$ for some $j\in V_i$. Then, for $\overline{S'}=\overline{S^*}\cup \{j\}$, we have $x(\overline{S'})< x(\overline{S^*})+2v(N)$, whilst $v(\overline{S'})\geq v(\overline{S^*})+2v(N)$. So, adding $j$ to $\overline{S^*}$ would preserve the violation. Therefore, if a coalition $\overline{S^*}$ is blocking, then it will remain blocking when making the above operations until we obtain a blocking coalition $\overline{S}= \bigcup_{i \in S} V_i \cup \{r_i\}$.}
In the case of $\overline{S}=\bigcup_{i \in S} V_i \cup \{r_i\}$, however, $\overline{v}(\overline{S}) =v(S) +\sum_{i \in S} 2v(N)|V_i|$ and $\overline{x}(\overline{S})=x(S)+\sum_{i \in S} 2v(N)|V_i|$, so that the core constraint follows from the fact that $x \in \mathit{core}(N,v)$, and thus $x(S) \ge v(S)$ holds.

Now suppose that $\overline{x}\in \mathit{core}(\overline{N},\overline{v})$. Then we can use the same arguments in reverse order to prove that $x \in \mathit{core}(N,v)$, as follows.
For every $S\subseteq N$, let $\overline{S}=
\bigcup_{i \in S} V_i \cup \{r_i\}$. Then
\[\begin{array}{lcl}
v(S) +\sum_{i \in S} 2v(N)|V_i| &= &\overline{v}(\overline{S})\\[3pt]
& \leq &\overline{x}(\overline{S})\\[3pt]
&=&x(S)+\sum_{i \in S} 2v(N)|V_i|,
\end{array}\]
which implies that $v(S)\geq x(S)$ for every $S\subseteq N$, thus $x \in \mathit{core}(N,v)$.
{Finally, if $\overline{x}\notin \mathit{core}(\overline{N},\overline{v})$ and one can find a blocking coalition $\overline{S}$ in $\overline{N}$, then by the above argument we can also find a blocking coalition $S$ for $x$ in $(N,v)$, completing the proof for P1.}

{The reduction of P3 is implied directly by the above argument for P1. The transformation of a core allocation $\overline{x}$ for $(\overline{N},\overline{v})$ into a core allocation $x$ for $(N,v)$ takes polynomial time. Hence, if we can find a core allocation for $(\overline{N},\overline{v})$ in polynomial time, then we can also find a core allocation for $(N,v)$ in polynomial time.}

\medskip
\noindent
To reduce P2, it suffices to show that $\mathit{core}(N,v)$ is non-empty if and only if $\mathit{core}(\overline{N},\overline{v})$ is non-empty. Above we already showed that if $\mathit{core}(N,v)$ is non-empty, then $\mathit{core}(\overline{N},\overline{v})$ is non-empty.
Hence, assume $\mathit{core}(N,v)= \emptyset$. Then, by the Bondareva-Shapley Theorem~\cite{Bo63,Sh67}, there is a function $\lambda: 2^N\setminus \{\emptyset\}\ \to [0,1]$ with for every $i\in N$, $\sum_{S\ni i} \lambda(S)=1$ such that $\sum_{S\neq \emptyset}\lambda(S)v(S)>v(N)$.
For every $S\subseteq N$, we define again $\overline{S}=
\bigcup_{i \in S} V_i \cup \{r_i\}$. Note that $\overline{v}(\overline{S})=v(S) +\sum_{i \in S} 2v(N)|V_i|$. We define $\overline{\lambda}(\overline{S})=\lambda(S)$ for every $S\subseteq N$ and let $\overline{\lambda}(S^*)=0$ for every other $S^*\subset \overline{N}$. First note that $\sum_{S^*\ni i} \overline{\lambda}(S^*)=1$ for every $i\in \overline{N}$. Furthermore, it holds that
\[\begin{array}{lll}
\sum_{S^*\subseteq \overline{N}}\overline{\lambda}(S^*)\overline{v}(S^*) &= &\sum_{{\emptyset \neq }S\subseteq N}\overline{\lambda}(\overline{S})\overline{v}(\overline{S})\\[8pt]
&= &\sum_{{\emptyset \neq }S\subseteq N}\lambda(S)\left(v(S) +\sum_{i \in S} 2v(N)|V_i|\right)\\[8pt]
&=&\sum_{{\emptyset \neq }S\subseteq N}\lambda(S)v(S) +\sum_{i \in N}\left(\sum_{S \ni i}\lambda(S)2v(N)|V_i|\right)\\[8pt]
&=&\sum_{{\emptyset \neq }S\subseteq N}\lambda(S)v(S) +\sum_{i \in N}2v(N)|V_i|\\[8pt]
&= &\sum_{{\emptyset \neq }S\subseteq N}\lambda(S)v(S) + 2v(N)\cdot|V|\\[8pt]
&> &v(N)+2v(N)\cdot|V|\\[8pt]
&=&\overline{v}(\overline{N}).
\end{array}\]
Therefore,
{by the Bondareva-Shapley Theorem, we have that}
$\mathit{core}(\overline{N},\overline{v})=\emptyset$. This completes the proof of Theorem~\ref{t-struc2}.
\qed
\end{proof}

\section{The Proofs of Theorems~\ref{t-easy1} and~\ref{t-easy2}}\label{s-results2}

In this section we give the proofs of Theorems~\ref{t-easy1} and~\ref{t-easy2}. To prove Theorem~\ref{t-easy1} we need an additional result as a lemma;
a similar construction was used by Plesnik~\cite{Pl99} to solve a constrained matching problem. Note that in {Lemma~\ref{l-interval}} we allow arbitrary edge weightings and thus ${\cal M}$ is the set of maximum weight matchings. Hence, {Lemma~\ref{l-interval}} is a slightly more general result than we strictly need.
{Note also that each interval $I_i$ in the statement of Lemma~\ref{l-interval} may either be open or closed. We need to allow both options in order to apply Lemma~\ref{l-interval} as a subroutine in our algorithm {\tt Lex-Min}, which we will present immediately after proving Lemma~\ref{l-interval}.}

\begin{lemma}\label{l-interval}
Given a partitioned matching game $(N,v)$ on a graph $G=(V,E)$ with a positive edge weighting $w$, partition ${\cal V}$ of $V$,
and intervals $I_1, \dots, I_n$, it is possible in $O(|V|^3)$-time to decide
if there exists a matching $M\in {\mathcal M}$ with $s_p(M)\in I_p$ for $p=1,\ldots,n$, and to find such a matching (if it exists).
\end{lemma}

\begin{proof}
First assume that the intervals $I_1,\ldots, I_n$ are closed.
For $p=1,\ldots,n$, we let $I_p = [a_p,b_p]$, where we assume without loss of generality that $b_p\le |V_p|$. We extend $(G,w)$ to a weighted graph $(\overline{G},\overline{w})$ in linear time as follows.  For $p=1, \dots ,n$, we add a set $B_p$ of $|V_p|-b_p$ new vertices,
 each of them joined to all vertices of $V_p$ by edges of weight $\overline{w}_e=0$. We also introduce a set $A_p$ of $b_p-a_p$ new vertices that are completely joined to all vertices of $V_p$ by edges of weight $\overline{w}_e=0$. In addition, all vertices in $\bigcup_p A_p$ are joined to each other by edges of weight $\overline{w}_e=0$. The original edges $e \in E$ in $G$ keep their (original) weights, that is, $\overline{w}_e=w_e$. If the total number of vertices is odd, we add an additional vertex $\overline{v}$ and join it by zero weight edges to all vertices of $\bigcup_p A_p$. This completes the description of $(\overline{G},\overline{w})$. Note that $|V(\overline{G})|\leq 2|V|$.

{First, we prove that there is a correspondence between the
matchings in $G$
satisfying the interval conditions  (which may not necessarily belong to ${\cal M}$) and the perfect matchings in $\overline{G}$.}

{In the one direction, suppose there exists a matching $M$ in $G$}
with $s_p(M)\in [a_p,b_p]$ for $p=1,\ldots,n$. As $s_p(M) \le b_p$, we can match all vertices of $B_p$ to $V_p$ by all zero weight edges. Finally, since $s_p(M) \ge a_p$, we can match the $b_p-s_p(M)\leq b_p-a_p$ remaining vertices in $V_p$ to vertices from  $A_p$. Thus all vertices of $V_p$ will be matched. If after doing this for every $p\in \{1,\ldots,n\}$, there still exist unmatched vertices in $\bigcup_p A_p$, then we match these to each other and, should their number be odd, to the extra vertex~$\overline{v}$. This yields a perfect matching $\overline{M}$ in $\overline{G}$.

In the other direction, let $\overline{M}$ be a perfect matching in $\overline{G}$.
Let $M := \overline{M} \cap E$ denote the corresponding matching in $G$. As $\overline{M}$ matches all vertices of $B_p$ into $V_p$, we know that $M$ leaves at least $|V_p| -b_p$ vertices unmatched. Hence, $s_p(M) \le b_p$ as required. Similarly, since all vertices of $V_p$ are matched by $\overline{M}$ and at most $|V_p|-b_p+b_p-a_p =|V_p|-a_p$ vertices in $V_p$ can be matched to $B_p \cup A_p$, we find that $M$ matches at least $a_p$ vertices in $V_p$. Hence, $s_p(M) \ge a_p$, as required.

Now, let $w^*$ be the maximum weight of a matching in $G$. Let $\overline{w}^*$ denote the maximum weight of a \emph{perfect} matching in $\overline{G}$, {if a perfect matching exists in $\overline{G}$}. Note that we can compute $w^*$ and~$\overline{w}^*$, and corresponding matchings, in $O(|V|^3)$ time~\cite{Ed65}.
{Note also that $w^*\geq \overline{w}^*$ holds, as all the new edges in $\overline{G}$ have weight $0$.}

{There are three possible cases, which we consider separately:}

\medskip
\noindent
{{\bf Case 1.} A perfect matching does not exist in $\overline{G}$.\\ Then, due to the correspondence we showed between the
matchings in $G$ satisfying the interval conditions and the perfect matchings in $\overline{G}$, there exists no
matching in $G$, and thus no maximum weight matching $M\in {\cal M}$, satisfying the interval conditions.}

\medskip
\noindent
{{\bf Case 2.} A perfect matching $\overline{M}$ exists in $\overline{G}$, and $w^*=\overline{w}^*$.\\
Above we deduced that $M := \overline{M} \cap E$ satisfies the interval conditions. As $w(M)=w(\overline{M})=\overline{w}^*=w^*$, we find that $M$ belongs to ${\cal M}$. Note that we can find $M$ in polynomial time (after we found $\overline{M}$ in polynomial time).}

\medskip
\noindent
{{\bf Case 3.} A perfect matching exists in $\overline{G}$, and $w^*>\overline{w}^*$.\\ This implies that there exists no maximum weight matching $M\in {\mathcal M}$ with $s_p(M)\in [a_p,b_p]$ for $p=1,\ldots,n$. For a contradiction, assume there exists a maximum weight matching $M\in {\mathcal M}$ satisfying the interval conditions. By the above correspondence, we now find a perfect matching~$\overline{M}$ in $\overline{G}$ with weight $w^*$, contradicting our assumption that $w^*>\overline{w}^*$.}

\medskip
\noindent
Now suppose we are given a set of intervals $I_1, \dots I_n$, some of which are open instead of closed. Let $I_p$ be an open interval. Recall that the $s$-values are sizes of subsets of matching edges and thus are integers. Hence, we may replace $I_p$ by  the largest closed interval with integer end-points contained in $I_p$ if this closed interval exists. If not, then a matching $M\in {\cal M}$ with $s_p(M)\in I_p$ does not exist.
\qed
\end{proof}

\noindent
Using Lemma~\ref{l-interval}, algorithm {\tt Lex-Min} computes for a partitioned matching game $(N,v)$,
{defined on a graph $G=(V,E)$ with vertex partition $(V_1,\ldots,V_n)$,}
and an allocation $x$,  values $d_1,\ldots, d_t$ with $d_1> \ldots > d_t$ for some integer $t\geq 1$, and, as we will prove, returns a {maximum} matching $M\in {\cal M}$ that is {strongly close to}~$x$.

\noindent \rule{\textwidth}{0.1mm}\\
\noindent
{\tt Lex-Min}\\[5pt]
\begin{tabular}[h]{lcl}
{\it input} &:& a partitioned matching game $(N,v)$ and an allocation $x$\\
{\it output} &:& a {maximum} matching $M\in {\cal M}$ that is {strongly close to} $x$.
\end{tabular}

\smallskip
\noindent
{\bf Step~1.} Compute the smallest number $d_1 \ge 0$ such that there exists a {maximum} matching $M \in \mathcal{M}$ with
$|x_p-s_p(M)|  \le d_1 \text{~~for all~~} p \in N$.

\smallskip
\noindent
{\bf Step~2.} Compute a minimal set $N_1 \subseteq N$ (with respect to set inclusion) such that there exists a {maximum} matching $M \in \mathcal{M}$ with\\[-15pt]
\begin{align}\nonumber
&|x_p-s_p(M)|
= d_1& \text{~~for all~~} p \in N_1\\
&|x_p-s_p(M)|  < d_1& \text{~~for all~~} p \in N\setminus N_1.\nonumber
\end{align}

\smallskip
\noindent
{\bf Step 3.} Proceed in a similar way for $t\geq 1$:
\begin{itemize}
 \item {\bf while} $N_1 \cup \dots \cup N_t \neq N$ {\bf do}
 \begin{itemize}
 \item $t \leftarrow t+1$.
 \item $d_t \leftarrow$ smallest $d$ such that there exists a {maximum} matching $M \in \mathcal{M}$ with
 \begin{align}\nonumber
 &|x_p-s_p(M)| =
  d_j &\text{~~for all~~} p \in N_j, ~j\le t-1\\
 &|x_p-s_p(M)|  \le d_t &\text{~~for all~~} p \in N\setminus (N_1 \cup \dots \cup\nonumber N_{t-1}).
 \end{align}
 \item $N_t \leftarrow$ inclusion minimal subset of $N\setminus (N_1 \cup \dots \cup N_{t-1})$ such that there exists a matching $M \in \mathcal{M}$ with
 \begin{align} \nonumber
 &|x_p-s_p(M)|
 =
  d_j &\text{~~for all~~} p \in N_j, ~j\le t-1\\ \nonumber
 &|x_p-s_p(M)|
=
 d_t &\text{~~for all~~} p \in N_t\\ \nonumber
 &|x_p-s_p(M)|  < d_t &\text{~~for all~~} p \in N\setminus (N_1 \cup \dots \cup N_{t}). \nonumber
 \end{align}
 \end{itemize}
\end{itemize}
\smallskip
\noindent
{\bf Step 4.} Return a {maximum} matching $M\in {\cal M}$ with
$|x_p-s_p(M)|=d_j$ for all $p\in N_j$ and all $j\in \{1,\ldots,t\}$.

\noindent \rule{\textwidth}{0.1mm}\\[2pt]

\noindent
We say that the countries in a set $N\setminus (N_1\cup \cdots \cup N_{t-1})$ are {\it unfinished} and that a country is {\it finished} when it is placed in some $N_t$. Note that {\tt Lex-Min} terminates as soon as all countries are finished.

We are now ready to prove Theorem~\ref{t-easy1}, which we restate below.

\medskip
\noindent
{\bf Theorem~\ref{t-easy1} (restated).}
{\it It is possible to find a {strongly close} maximum matching for a uniform partitioned matching game $(N,v)$ and target allocation $x$ in polynomial time.}

\begin{proof}
To prove the theorem we will show that the {\tt Lex-Min} algorithm is correct {(that is, the matching returned in Step~4 is strongly close to~$x$)} and runs in
{$O(n|V|^3(\log|V|+n))$}
time for a uniform partitioned matching game $(N,v)$ with an allocation $x$.

\medskip
\noindent
{\it Correctness proof.}
We first prove the correctness of {\tt Lex-Min}. Let $(N,v)$ be a uniform partitioned matching game on a graph $G=(V,E)$ with partition ${\cal V}$, and let $x$ be an allocation. Let $\overline{M}$ be the matching from ${\cal M}$ that is returned by {\tt Lex-Min}.
We claim that $\overline{M}$ is {strongly close to} $x$.
In order to prove this, let  $M^* \in \mathcal{M}$ be a matching {that is strongly close to $x$}. Since both $\overline{M}$ and $M^*$ are maximum matchings, we have
$M^* = \overline{M} \oplus \mathcal{P} \oplus \mathcal{C}$,
where $\mathcal{P}$ and $\mathcal {C}$ are sets of even alternating paths and (even) alternating cycles, respectively. We make the additional assumption that among all the {strongly close} matchings in $\mathcal{M}$, the matching $M^*$ is chosen as \emph{closest} to $\overline{M}$ in the sense that $|\mathcal{P}|+|\mathcal{C}|$ is as small as possible.

We claim that $\mathcal{C} =\emptyset$. Otherwise, we would switch $M^*$ to another maximum matching along an alternating cycle $C\in {\cal C}$. This yields a new {maximum} matching $M^* \oplus C \in \mathcal{M}$, which is again {strongly close to $x$}, as the switch does not affect any $s_p(M^*)$. Moreover, $M^* \oplus C$ is closer to $\overline{M}$, contradicting our choice of $M^*$.
Hence, there exists a disjoint union of paths $P_1,\ldots,P_k$ such that
$M^* = \overline{M} \oplus P_1 \oplus \dots \oplus P_k$.
We claim that each $P_j$ has endpoints in different countries; otherwise switching from $M^*$ to $\overline{M}$ along $P_j$ would not affect any $s_p(M^*)$ and so again leads to a new {strongly close maximum} matching closer to $\overline{M}$.

Let $d_1^* > d_2^* > \dots $ denote the different values of $|x_p-s_p(M^*)|$ and let $N_j^* \subseteq N$ denote the corresponding sets of players $p \in N$ with $|x_p-s_p(M^*)|=d_j^*$. We prove by induction on $t$ that for every $t$, it holds that $d_t^*=d_t$ and $N_t^*=N_t$, which implies that
{$d(M^*) = d(\overline{M})$} and thus $\overline{M}$ is {strongly close to $x$}. Let $t=1$.
In Claims 1 and~2, we prove that $d^*_1=d_1$ and $N_1^*=N_1$.

\medskip
\noindent
\emph{Claim 1}: $d^*_1=d_1$.\\
\emph{Proof}:
As $M^*$ is {strongly close to $x$, we have that}  $d^*_1 \le d_1$. If $d^*_1 < d_1$, then $d_1$ was not chosen as small as possible,
since $M=M^*$ satisfies $|x_p-s_p(M)| \le d^*_1$ for every $p \in N$.  Hence
$d^*_1 = d_1$. \dia

\medskip
\noindent
\emph{Claim 2}: $N^*_1 = N_1$.\\
\emph{Proof}: If $N^*_1 \subsetneq N_1$, then $N_1$ is not minimal, as
$|x_p-s_p(M^*)|  \le d^*_1=d_1$ for $p \in N^*_1$ and
$|x_p-s_p(M^*)|  < d^*_1=d_1$ for $p \in N\setminus N^*_1$. This
contradicts Step 2 of {\tt Lex-Min}.
Hence,
{$N_1^* \not \subset N_1$}.
{For a contradiction, l}et
$p_0 \in N^*_1\setminus N_1$. So,
$|x_{p_0}-s_{p_0}(M^*)|  = d^*_1=d_1 > |x_{p_0}-s_{p_0}(\overline{M})|$.
We distinguish two cases.

\medskip
\noindent
{\bf Case 1.} $s_{p_0}(M^*) =  x_{p_0}+d^*_1$.\\
Then $s_{p_0}(M^*) > s_{p_0}(\overline{M})\geq 0$, so there exists an even path, say, $P_1$ starting in $V_{p_0}$ with an $M^*$-edge and ending in some $V_{p_1}$ with an $\overline{M}$-edge. Recall that the endpoints of $P_1$ are in different countries, hence $p_1 \neq p_0$.
Note that replacing $M^*$ by $M^* \oplus P_1$ would decrease $s_{p_0}(M^*)$ and increase $s_{p_1}(M^*)$.

Assume first that $s_{p_1}(M^*) \ge s_{p_1}(\overline{M}) $. As $P_1$ ends in $V_{p_1}$, we have $s_{p_1}(\overline{M}) \ge 1 $. Hence,  there exists an alternating path, say $P_2$, that starts with an $M^*$-edge in~$V_{p_1}$
and ends with an $\overline{M}$-edge in some $V_{p_2}$, $p_2 \neq p_1$. If $p_2=p_0$, then $M^* \oplus P_1 \oplus P_2$ would be {strongly close to $x$} and closer to $\overline{M}$, a contradiction. Hence $p_2 \notin \{p_0,p_1\}$ and in case $s_{p_2}(M^*) \ge s_{p_2}(\overline{M})$, we may continue to construct a sequence $p_0, p_1, p_2, ...$ in this way. Note that whenever we run into a cycle, that is, when $p_s = p_r$ for some $r>s$, then we get a contradiction
by observing that $M^* \oplus  P_{r+1} \oplus \dots \oplus P_s$ is also {strongly close to $x$} (indeed switching $M^*$
to $\overline{M}$ along $P_{r+1}, \dots, P_s$ does not affect any $s_p(M^*)$) and closer to $\overline{M}$. Hence,
eventually, our sequence $p_0, p_1, \dots, p_r$ must end up with $s_{p_r}(M^*) < s_{p_r}(\overline{M})$
 {for some $r\geq 1$}.

We {now} derive that $s_{p_r}(M^* \oplus P_1 \oplus \dots \oplus P_r) \le s_{p_r}(\overline{M}) \le x_{p_r}+d_1$.
If $p_r \notin N_1$, then we even get $s_{p_r}(M^* \oplus P_1 \oplus \dots \oplus P_r) \le s_{p_r}(\overline{M}) < x_{p_r}+d_1$.
We now define the matching $M':= M^* \oplus P_1 \oplus \dots \oplus P_r$.
Note that $M'$ is a maximum matching closer to $\overline{M}$. To obtain a contradiction with our choice of $M^*$ it remains to show that $M'$ is {strongly close to $x$}.

We first consider $p=p_r$.
We have $s_{p_r}(M') = s_{p_r}(M^*)+1 \ge x_{p_r} -d_1+1 > x_{p_r} -d_1$. Combining this with the upper bound found above, we obtain
\begin{eqnarray}\label{eq-bound_1}
x_{p_r} - d_1 < s_{p_r}(M') \le x_{p_r}+d_1 ~\text{if} ~ p_r\in N_1\\
x_{p_r} - d_1 < s_{p_r}(M') < x_{p_r}+d_1 ~\text{if} ~p_r \notin N_1. \nonumber
\end{eqnarray}
For $p=p_0$, we have
that $s_{p_0}(M') = s_{p_0}(M^*)-1 = x_{p_0}+d_1-1$, where the last equality holds because $p_0\in N_1^*$. Hence, we found that
\vspace*{-0.2cm}
\begin{equation}\label{eq-bound_2}
|s_{p_0}(M')-x_{p_0}|=d_1-1.
\end{equation}
From (\ref{eq-bound_1}) and (\ref{eq-bound_2}) and the fact that $|s_p(M')-x_p|=|s_p(M^*)-x_p|$ if $p\notin \{p_0,p_r\}$,
we conclude that either $M'$ is
{lexicographically smaller than $M^*$},
which yields {a} contradiction
{(as $M^*$ is strongly close to $x$)},
or  $p_r \in N_1$ and $s_{p_r}(M') = x_{p_r} +d_1$.
Assume that the latter case holds. Then we have  $$s_{p_r}(M^*) = s_{p_r}(M')-1= x_{p_r}+d_1-1.$$ However, then $M^*$ and $M'$ are symmetric with respect to $p_0$ and $p_r$, namely,
$|s_{p_0}(M^*)-x_{p_0}|=d_1=|s_{p_r}(M')-x_{p_r}|$ and
$|s_{p_r}(M^*)-x_{p_r}|={d_1-1}=|s_{p_0}(M')-x_{p_0}|$.
Combining these equalities with the fact that $|s_p(M')-x_p|=|s_p(M^*)-x_p|$ if $p\notin \{p_0,p_r\}$ implies that $M'$ is {strongly close to $x$}, our desired contradiction. Hence,  $N^*_1\setminus N_1=\emptyset$. As
{$N_1^* \not \subset N_1$},
we conclude that $N_1^*= N_1$.

\medskip
\noindent
{\bf Case 2.}  $s_{p_0}(M^*) =  x_{p_0}-d^*_1$.\\
In this case we have $s_{p_0}(M^*) < s_{p_0}(\overline{M})$. Hence, there must exist an alternating path $P_1$ that starts in $V_{p_0}$ with an $\overline{M}$-edge and that ends in some $V_{p_1}$ with an $M^*$-edge. Just as in Case~1, it holds that $p_1\neq p_0$. If $s_{p_1}(M^*) \le s_{p_1}(\overline{M})$, we may
continue with an alternating path $P_2$ starting from $V_{p_1}$ and ending in some $V_{p_2}$ with $p_2 \notin \{p_0, p_1\}$, and continuing in this way we eventually end up with a sequence $p_0, p_1, \dots, p_r$ such that $ s_{p_r}(M^*) > s_{p_r}(\overline{M})$. Then $M':= M^* \oplus P_1 \dots \oplus P_r$ has $s_{p_0}(M') = s_{p_0}(M^*) +1$ and
$s_{p_r}(M') = s_{p_r}(M^*) -1$. By the same arguments that we used in Case~1, we prove that $M'$ is a maximum matching that is {strongly close to $x$} and that is closer to $\overline{M}$ than $M^*$ is. This contradicts our choice of $M^*$, and we have proven Claim~2. \dia

\medskip
\noindent
Now let $t\geq 2$. Assume that  $d^*_1=d_1, \dots, d^*_{t-1} =d_{t-1}$ and $N^*_1=N_1, \dots, N^*_{t-1}=N_{t-1}$. By using the same arguments as in the proof of Claim~1, we find that $d^*_t=d_t$.
By using  the same arguments as in the proof of Claim~2, we find that
{$N_t^* \not \subset N_t$}.

We will now show that $N^*_t=N_t$. {For a contradiction,} consider a country $p_0 \in N^*_t \setminus N_t$ and split the proof into two cases similar to Cases~1 and~2 for the case $t=1$, namely when $s_{p_0}(M^*) =  x_{p_0}+d^*_1$ and  $s_{p_0}(M^*) =  x_{p_0}-d^*_1$.
We will only show the first case in detail, as the proof of the other case is similar.
Hence, from now on we assume that $s_{p_0}(M^*) =  x_{p_0}+d^*_t$.

We know that $p_0 \in N_t^*$, so $p_0 \notin N^*_1 \cup \dots \cup N^*_{t-1} = N_1 \cup \dots \cup N_{t-1}$. Hence $s_{p_0}(\overline{M}) \le x_{p_0}+d_t$, and $p_0 \notin N_t$ implies $s_{p_0}(\overline{M}) < x_{p_0}+d_t$, so that $s_{p_0}(M^*) > s_{p_0}(\overline{M})$. So there is an alternating path $P_1$ starting in $V_{p_0}$ with an $M^*$-edge and ending in some $V_{p_1}$ with an $\overline{M}$-edge. Again, we may assume that $p_1\neq p_0$.
If $s_{p_1}(M^*) \ge s_{p_1}(\overline{M})$, then there must be some alternating $P_2$ starting in $V_{p_1}$ with an $M^*$-edge and leading to some $V_{p_2}$ with $p_2 \notin \{p_0,p_1\}$ and so on, until eventually we obtain
a sequence $p_0, p_1, \dots, p_r$ with $s_{p_r}(M^*) < s_{p_r}(\overline{M})$ {for some $r\geq 1$}.

As before, we let $M':= M^* \oplus P_1 \dots \oplus P_r$ and note that $M'\in {\cal M}$ is closer to $\overline{M}$ than $M^*$ is. Hence, to obtain a contradiction, it remains to show that $M'$ is {strongly close to $x$}.
As $p_0 \notin N_1 \cup \dots \cup N_t$, we find that $s_{p_0}(M')=s_{p_0}(M^*) -1 \ge s_{p_0}(\overline{M}) > x_{p_0}-d_t$. On the other hand, $s_{p_0}(M')<s_{p_0}(M^*)=x_{p_0}+d_t$. Hence
\begin{equation}\label{eq-bound_2a}
|s_{p_0}(M') - x_{p_0}| < d_t.
 \end{equation}
Now consider $p=p_r$. We first rule out that $p_r \in N_1 \cup \dots \cup N_{t-1}$.
Assume to the contrary that $p_r \in N_j$ for some $j \in \{1, \dots, t-1\}$. Then as lower bound we have $s_{p_r}(M')= s_{p_r}(M^*)+1 \ge x_{p_r}-d_j+1 > x_{p_r}-d_j$, and as upper bound, $s_{p_r}(M') =s_{p_r}(M^*)+1 \le s_{p_r}(\overline{M}) \le x_{p_r}+d_j$. Hence,
\begin{equation}\label{eq-bound_3}
|s_{p_r}(M') -x_{p_r}| \le d_j.
\end{equation}
Inequality (\ref{eq-bound_3}), together with (\ref{eq-bound_2a}) and the fact that $|s_p(M') -x_{i}|=|s_p(M^*) -x_p|$ for $p\notin \{p_0,p_r\}$ shows
that $M'$ is lexicographically smaller than $M^*$, a contradiction. We conclude that $p_r \notin N_1 \cup \dots \cup N_{t-1}$.

We now have $s_{p_r}(\overline{M}) \le x_{p_r}+d_t$ if $p_r\in N_t$ and $s_{p_r}(\overline{M}) < x_{p_r}+d_t$ if $p_r\notin N_t$. Hence, we can repeat the arguments that we used for the case where $t=1$ to obtain our contradiction. This completes the correctness proof of {\tt Lex-Min}.

\smallskip
\noindent
{\it Running time analysis.}
As $x_p$ is fixed and $s_p(M)$ is an integer between $0$ and
{$|V|$},
there are $O(|V|)$ values for $|x_p-s_p(M)|$. Hence, we can find $d_t$ by binary search in $O(\log |V|)$ time.
This requires $O(\log |V|)$ applications of Lemma~\ref{l-interval}, each of which taking time $O(|V|^3)$. So, finding a $d_t$ takes $O(|V|^3\log|V|)$ time.
{For each $d_t$, we compute an inclusion minimal set $N_t$ as follows. For each unfinished country $q$, we check if there exists a maximum matching $M\in {\cal M}$ such that
\begin{align} \nonumber
 &|x_p-s_p(M)|
 =
  d_j &\text{~~for all~~} p \in N_j, ~j\le t-1\\ \nonumber
 &|x_p-s_p(M)|
\leq
 d_t &\text{~for all unfinished~~} p\neq q\\ \nonumber
 &|x_q-s_q(M)|  < d_t. \nonumber
 \end{align}
We can do this in $O(|V|^3)$ time by applying Lemma~\ref{l-interval}. If this is possible, then $q$ does not belong to~$N_t$, and else we put $q$ in $N_t$.
As there are at most $n$ countries to consider, we need to do this at most $n$ times.}
Each time we find a $d_t$, the number of finished countries increases by at least one.  Hence, as the number of countries is $n$, we find that the total running time of {\tt Lex-Min} is
{$O(n|V|^3(\log|V|+n))$}.
\qed
\end{proof}

\noindent
We now give the proof of Theorem~\ref{t-easy2}. In our proof we use exponentially increasing weights to minimize deviations from a target solution in a lexicographic way. This is similar in nature to techniques used in the literature for representing lexicographic preferences of agens over bundles in many-to-one allocation problems, see, for example,~\cite{ABLLM19}.

\medskip
\noindent
{\bf Theorem~\ref{t-easy2} (restated).}
{\it  It is possible to find a {strongly close} maximum weight matching for a directed partitioned matching game $(N,v)$ of width~$1$ {(i.e., for a directed matching game)} and target allocation $x$ in polynomial time.}

\begin{proof}
	Assume that our directed compatibility graph is $\overline{G}=(V,A)$, and our undirected compatibility graph is $G=(V,E)$, where $ij\in E$ if both $(i,j)$ and $(j,i)$ are arcs in $A$.
	Extend $G$ to a complete graph by adding zero-weight edges to it; note that a {strongly close} maximum weight matching in the extended graph yields a {strongly close} maximum weight matching in $G$ if we forget its newly added zero-weight edges.
	To keep notation simple, from now on we will assume that $G$ is a complete graph.
	
	Furthermore, if $|V|$ is odd, then we add a dummy vertex $v$ to $G$ with target allocation $x_v=0$.
	This still does not change the structure of the {strongly close} maximum weight matchings, however, now each of them can be extended to a {strongly close} maximum weight \emph{perfect} matching by adding some zero-weight edges to it.
	Because of this, from now on we focus on this problem in a weighted complete graph with an even number of vertices.
	
	For a target allocation $x$, we create an edge weighting $w^x$ of $G$ where the weight of each edge will be a vector from $\R^2$, that is, $w^x(e)=(\alpha,\beta)$, where $w_1^x(e)=\alpha$ and $w_2^x(e)=\beta$ are both real numbers.
	We can add or subtract two vectors coordinatewise, and compare them lexicographically, i.e., $w<w'$ if $w_1<w_1'$ or $w_1=w_1'$ and $w_2<w_2'$.
	
	For every edge $pq\in E$, we let $w_{1}^x(pq)=w_{pq}+w_{qp}$.
	This is to ensure that the maximum weight matching according to $w^x$ gives a maximum weight
{matching}
according to $w$.
	
	To define $w_{2}^x(pq)$, we first let $\delta^x_{pq}=|x_p-w_{qp}|$ for all ordered pairs $p,q\in V$.
	We introduce a vector $\Delta^x$ that contains the $n(n-1)$ values of $\delta^x$ in a weakly increasing order.	
	Let $r_{pq}$ denote the (average) rank of $\delta^x_{pq}$ in $\Delta^x$, meaning that if there is a tie from position $k$ to position $l$ in $\Delta^x$, then each of these elements will have rank $\frac{k+l}{2}$.
	Define $w_{2}^x(pq)=-2^{r_{pq}}-2^{r_{qp}}$ for all $pq\in E$.
	
	The $w$ weight of a matching $M$ is
	$\sum_{pq\in M} w_{pq}+w_{qp}=\sum_{pq\in M} w_{1}^x(pq).$
	If this latter sum is maximized by $M$, we can conclude that $M$ has maximum weight.
	
	Similarly, let $d(M)= (|x_{p_1}-u_{p_1}(M)|, \dots, |x_{p_n}-u_{p_n}(M)|)$ be the deviation vector of the perfect matching $M$ obtained by reordering the components $|x_p-u_p(M)|$ non-increasingly.
	If the deviation vector {$d(M')$} of some other perfect matching $M'$ {is} lexicographically larger than $d(M)$, {then we would have that} $$\sum_{pq\in M'} w_{2}^x(pq)<\sum_{pq\in M} w_{2}^x(pq)=\sum_{p\in V, pq\in M} -2^{r_{pq}}.$$
	{To see why the first inequality holds}, the first few terms in $\sum_{p\in V, pq\in M'} -2^{r_{pq}}$ are the same as in $\sum_{p\in V, pq\in M} -2^{r_{pq}}$, until we reach a term that is {strictly} larger in $M$, and by the exponential growth of $2^{r_{pq}}$ this single term is even larger than all further terms {together} in $M'$.
	Thus, if $\sum_{pq\in M} {w_{2}^x(pq)}$ is maximized by $M$ among some perfect matchings, then we can conclude that $M$ is {strongly close} among them.
	
	We can use this for the perfect matchings that maximize $\sum_{pq\in M} w_{1}^x(pq)$, to conclude that a maximum weight perfect matching according to $w^x$ is a {strongly close} maximum weight perfect matching according to $w$.
	Since we can find a maximum weight perfect matching according to $w^x$ by running the blossom algorithm of Edmonds \cite{Ed65} whose running time is $O(n^3)$ (independent of the weights\footnote{Note that these algorithms indeed work in any ordered abelian group; see this short argument by Emil Je\v r\'abek: \url{https://cstheory.stackexchange.com/a/52389/419}.}), we are done.\qed
\end{proof}

\section{The Proofs of Theorems~\ref{t-hard1}--\ref{t-hard4}}\label{s-results3}

In this section we prove Theorems~\ref{t-hard1}--\ref{t-hard4}.
For the first three theorems we will reduce from one of the following two problems. The problem {\sc Partition} is well-known to be \NP-complete~\cite{GJ79}, and has as input a set of $k$ integers $a_1, \dots, a_k$. The question is whether there exists a set of indices $I \subseteq \{1, \dots, k\}$ with $$a(I)=\frac{1}{2}\sum_{i=1}^k a_i.$$
The problem {\sc $3$-Partition} is to decide if we can partition a set of $3k$ positive integers $a_1,\ldots,a_{3k}$ with $\sum_{p=1}^{3k}a_p=kc$, for some integer~$c$, into $k$ sets that each sum up to $c$.
In contrast to {\sc Partition}, the {\sc $3$-Partition} problem is even strongly \NP-complete~\cite{GJ79} (so \NP-complete when encoded in unary) even if $\frac{1}{4}c<a_i<\frac{1}{2}c$. The latter property ensures that each set in a solution has size exactly~$3$.

We can now start with giving the proofs, the first one of which is the proof of Theorem~\ref{t-hard1}.

\medskip
\noindent
{\bf Theorem~\ref{t-hard1} (restated).}
{\it It is \NP-hard to find a {weakly close} maximum weight matching for a directed partitioned matching game~$(N,v)$ with $n=2$ and target allocation~$x$.}

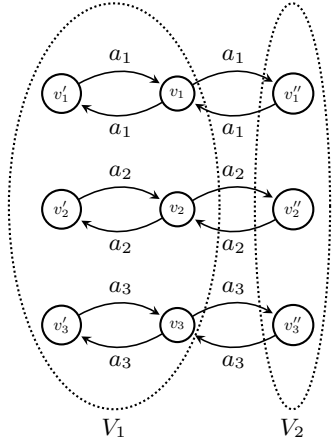
\begin{figure}[t]
\begin{center}
 \begin{tikzpicture}[
            > = stealth, 
            shorten > = 1pt,
            auto,
            node distance = 2.2cm, 
            semithick 
        ]

        \tikzstyle{every state}=[scale=0.7,
            draw = black,
            thick,
            fill = white,
            minimum size = 3mm
        ]
        \node[state] (v_1') {$v_1'$};
        \node[state] (v_1) [right of=v_1'] {$v_1$};
        \node[state] (v_1'') [right of=v_1] {$v_1''$};
        \path[->] [bend left] (v_1') edge node [above] {$a_1$} (v_1);
        \path[->] [bend left] (v_1) edge node [below] {$a_1$} (v_1');
         \path[->] [bend left] (v_1) edge node [above] {$a_1$} (v_1'');
        \path[->] [bend left] (v_1'') edge node [below] {$a_1$} (v_1);
        \node[state] (v_2') [below of=v_1'] {$v_2'$};
        \node[state] (v_2) [right of=v_2'] {$v_2$};
        \node[state] (v_2'') [right of=v_2] {$v_2''$};
        \path[->] [bend left] (v_2') edge node [above] {$a_2$} (v_2);
        \path[->] [bend left] (v_2) edge node [below] {$a_2$} (v_2');
         \path[->] [bend left] (v_2) edge node [above] {$a_2$} (v_2'');
        \path[->] [bend left] (v_2'') edge node [below] {$a_2$} (v_2);
        \node[state] (v_3') [below of=v_2'] {$v_3'$};
        \node[state] (v_3) [right of=v_3'] {$v_3$};
        \node[state] (v_3'') [right of=v_3] {$v_3''$};
        \path[->] [bend left] (v_3') edge node [above] {$a_3$} (v_3);
        \path[->] [bend left] (v_3) edge node [below] {$a_3$} (v_3');
         \path[->] [bend left] (v_3) edge node [above] {$a_3$} (v_3'');
        \path[->] [bend left] (v_3'') edge node [below] {$a_3$} (v_3);

        \draw[densely dotted,thick] (3.07,-1.5) ellipse (0.5 and 2.7) node[below,yshift=-2.7cm]{$V_{2}$};
        \draw[densely dotted,thick] (0.7,-1.5) circle (1.4 and 2.7) node[below,yshift=-2.7cm]{ $V_{1}$};
    \end{tikzpicture}
\end{center}
\caption{The construction of the (directed) compatibility graph in the proof of Theorem~\ref{t-hard1} for $k=3$.}\label{fig:thm6}
\end{figure}

\begin{proof}
We reduce from {\sc Partition}. From an instance $(a_1,\ldots,a_k)$ of {\sc Partition} we construct a partitioned matching game $(N,v)$ with $n=2$. We define $V_1=\{v_1, \dots, v_k, v_1', \dots, v_k'\}$ and $V_2= \{v_1'', \dots, v_k''\}$. For $i=1, \dots, k$ we have arcs $(v_i,v_i')$, $(v_i',v_i)$, $(v_i, v_i'')$ and $(v_i'',v_i)$, each with weight $a_i$. Every maximum weight matching $M$ matches each $v_i$ with either $v_i'$ or $v_i''$. Matching $v_i$ with $v_i'$ adds $2a_i$ to $u_1$
(and $0$ to $u_2$), while matching $v_i$ with $v_i''$ adds $a_i$ to both $u_1$ and $u_2$. Note that $v(N)=2\sum_j a_j$. Let $x$ be the allocation with $x_1=\frac{3}{2} \sum_j a_j$ and $x_2=\frac{1}{2} \sum_j a_j$.
Then there exists a {maximum weight} matching $M \in \mathcal{M}$ with $u_1(M) = x_1$ and $u_2(M) = x_2$ if and only if $(a_1,\dots,a_k)$ is a yes-instance of {\sc Partition}. {If such a maximum weight matching $M$ exists, then $M$ is also weakly close to $x$, since $\max_{p \in N}\{|x_p-u_p(M)|\} = 0$. Thus, if we could find a weakly close maximum weight matching in polynomial time, we can also decide {\sc Partition} in polynomial time.} \qed
\end{proof}

\noindent
We now prove Theorem~\ref{t-hard2}.

\medskip
\noindent
{\bf Theorem~\ref{t-hard2} (restated).}
{\it  It is \NP-hard to find a {weakly close} maximum weight matching for a $2$-sparse directed partitioned matching game $(N,v)$
and target allocation~$x$.}

\begin{proof}
We reduce from {\sc $3$-Partition}.
From an instance $(a_1,\ldots, a_{3k})$, such that $\sum_{p=1}^{3k}a_p=kc$ for some integer~$c$ and $\frac{1}{4}c<a_i<\frac{1}{2}c$, we construct a directed partitioned matching game $(N,v)$ on a (directed) compatibility graph $\overline{G}=(V,A)$ as follows (see also Fig.~\ref{fig:thm7}):

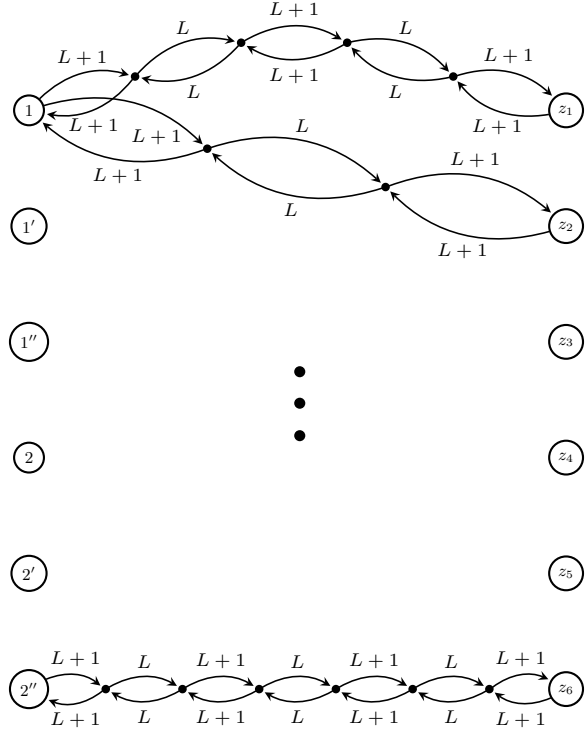
\begin{figure}[t]
\begin{center}
 \begin{tikzpicture}[
            > = stealth, 
            shorten > = 1pt, 
            auto,
            node distance = 2.2cm, 
            semithick 
        ]

        \tikzstyle{every state}=[scale=0.7,
            draw = black,
            thick,
            fill = white,
            minimum size = 3mm
        ]
        \node[state] (1) {$1$};
        \node[state] (1') [below of=1] {$1'$};
        \node[state] (1'') [below of=1'] {$1''$};
        \node[state] (2) [below of=1''] {$2$};
        \node[state] (2') [below of=2] {$2'$};
        \node[state] (2'') [below of=2'] {$2''$};
        \node[state] (z_1) [right of=1, xshift = 8cm] {$z_1$};
        \node[state] (z_2) [below of=z_1] {$z_2$};
        \node[state] (z_3) [below of=z_2] {$z_3$};
        \node[state] (z_4) [below of=z_3] {$z_4$};
        \node[state] (z_5) [below of=z_4] {$z_5$};
        \node[state] (z_6) [below of=z_5] {$z_6$};

        \node[fill,circle, black, scale=0.3,draw] (p1) [right of=1, xshift=2.5cm, yshift=1.5cm]{};
        \node[fill,circle, black, scale=0.3,draw] (p2) [right of=p1, xshift=2.5cm,yshift=1.5cm]{};
        \node[fill,circle, black, scale=0.3,draw] (p3) [right of=p2, xshift=2.5cm]{};
        \node[fill,circle, black, scale=0.3,draw] (p4) [right of=p3, xshift=2.5cm,yshift=-1.5cm]{};
        \path[->] [bend left] (1) edge node [above] {\scriptsize $L+1$} (p1);
        \path[->] [bend left] (p1) edge node [below] {\scriptsize $L+1$} (1);
        \path[->] [bend left] (p1) edge node [above] {\scriptsize $L$} (p2);
        \path[->] [bend left] (p2) edge node [below] {\scriptsize $L$} (p1);
        \path[->] [bend left] (p2) edge node [above] {\scriptsize $L+1$} (p3);
        \path[->] [bend left] (p3) edge node [below] {\scriptsize $L+1$} (p2);
        \path[->] [bend left] (p3) edge node [above] {\scriptsize $L$} (p4);
        \path[->] [bend left] (p4) edge node [below] {\scriptsize $L$} (p3);
        \path[->] [bend left] (p4) edge node [above] {\scriptsize $L+1$} (z_1);
        \path[->] [bend left] (z_1) edge node [below] {\scriptsize $L+1$} (p4);

                \node[fill,circle, black, scale=0.3,draw] (q1) [right of=1, xshift=5.7cm, yshift=-1.7cm]{};
        \node[fill,circle, black, scale=0.3,draw] (q2) [right of=q1, xshift=5.7cm, yshift=-1.7cm]{};
        \path[->] [bend left] (1) edge node [below, xshift=.35cm, yshift=-0.24cm] {\scriptsize $L+1$} (q1);
        \path[->] [bend left] (q1) edge node [below] {\scriptsize $L+1$} (1);
        \path[->] [bend left] (q1) edge node [above] {\scriptsize $L$} (q2);
        \path[->] [bend left] (q2) edge node [below] {\scriptsize $L$} (q1);
        \path[->] [bend left] (q2) edge node [above] {\scriptsize $L+1$} (z_2);
        \path[->] [bend left] (z_2) edge node [below] {\scriptsize $L+1$} (q2);

        \node[fill,circle, black, scale=0.3,draw] (r1) [right of=2'', xshift=1.2cm]{};
        \node[fill,circle, black, scale=0.3,draw] (r2) [right of=r1, xshift=1.2cm]{};
        \node[fill,circle, black, scale=0.3,draw] (r3) [right of=r2, xshift=1.2cm]{};
        \node[fill,circle, black, scale=0.3,draw] (r4) [right of=r3, xshift=1.2cm]{};
        \node[fill,circle, black, scale=0.3,draw] (r5) [right of=r4, xshift=1.2cm]{};
        \node[fill,circle, black, scale=0.3,draw] (r6) [right of=r5, xshift=1.2cm]{};
        \path[->] [bend left] (2'') edge node [above] {\scriptsize $L+1$} (r1);
        \path[->] [bend left] (r1) edge node [below] {\scriptsize $L+1$} (2'');
        \path[->] [bend left] (r1) edge node [above] {\scriptsize $L$} (r2);
        \path[->] [bend left] (r2) edge node [below] {\scriptsize $L$} (r1);
        \path[->] [bend left] (r2) edge node [above] {\scriptsize $L+1$} (r3);
        \path[->] [bend left] (r3) edge node [below] {\scriptsize $L+1$} (r2);
        \path[->] [bend left] (r3) edge node [above] {\scriptsize $L$} (r4);
        \path[->] [bend left] (r4) edge node [below] {\scriptsize $L$} (r3);
        \path[->] [bend left] (r4) edge node [above] {\scriptsize $L+1$} (r5);
        \path[->] [bend left] (r5) edge node [below] {\scriptsize $L+1$} (r4);
        \path[->] [bend left] (r5) edge node [above] {\scriptsize $L$} (r6);
        \path[->] [bend left] (r6) edge node [below] {\scriptsize $L$} (r5);
        \path[->] [bend left] (r6) edge node [above] {\scriptsize $L+1$} (z_6);
        \path[->] [bend left] (z_6) edge node [below] {\scriptsize $L+1$} (r6);

        \node[black, scale=3] (dots) [right of=1'', xshift=-1cm, yshift=-0.17cm]{\Large $\vdots$};
    \end{tikzpicture}
\end{center}
\caption{The construction of the (directed) compatibility graph $\overline{G}=(V,A)$ in the proof of Theorem~\ref{t-hard2} when $k=2$, $a_1=3$, $a_2=2$, $a_6=4$. For clarity reasons, only three paths are displayed and the other 33 paths between sources and sinks have not been drawn.}\label{fig:thm7}
\end{figure}

\begin{itemize}
\item We start with $3k$ \emph{sources}. That is, for $p =1, \dots, k$, we define $S_p:=\{p, p',p''\}$ and $S:= \bigcup_p S_p$.
\item  We add a set of $3k$ \emph{sinks}  $T:= \{z_1, \dots, z_{3k}\}$.
\item We join every source to every sink by a path. That is, from each $p$ (resp. $p'$ and $p''$) there is a path $P_{pq}$ (resp. $P_{p'q} $ and $P_{p''q}$) to each $z_q$ of (odd) length $2a_q-1$. This gives a total number of $(3k)^2$ pairwise internally vertex disjoint paths.
\item Every two consecutive vertices on each path are joined by two opposite arcs of equal weight. The weights on each path alternate between $L$ and $L+1$, starting and ending with opposite arcs of weight $L+1$, where $L$ is a sufficiently large integer, say, $L>kc$.
\item For $p=1,\ldots,k$, let $V_p= (\bigcup_q V((P_{pq} )\cup V(P_{p'q})\cup V(P_{p''q})))\setminus T$, and let
$V_{k+1}= T$.
\end{itemize}

\noindent
Note that the edges of the underlying graph $G=(V,E)$ have either weight $2L$ or $2(L+1)$. As $L >kc$, every maximum weight matching in $G$ is perfect. More precisely, every maximum weight matching $M$ of $G$ looks as follows. Each edge has either weight $2L$ or weight $2L+2$.
For $p=1, \dots, k$ there is a path $P_{pq}$ from $p$ to some $z_q$,
a path~$P_{p'q'}$ from $p'$ to some $z_{q'}$, and a path $P_{p''q''}$ from $p''$ to some $z_{q''}$ that are completely matched in the sense that $M \cap P_{pq}$ is a perfect matching of $P_{pq}$, and the same holds for $P_{p{'}q{'}}$ and $P_{p{''}q{''}}$. These paths contribute to $u_p(M)$ a total gain of
$$(2(a_q-1)+1)(L+1)+(2(a_{q'}-1)+1)(L+1)+(2(a_{q''}-1)+1)(L+1)=
(2(a_q+a_{q'}+a_{q''})-3)(L+1).$$
For $p=1,\ldots,k$,
there are also $3(k-1)$ paths from $\{p,p',p''\}$ to the remaining $3k-3$ sinks in $T\setminus \{z_q,z_{q'},z_{q''}\}$ that start and end with a non-matching edge (and are otherwise $M$-alternating). These paths contribute to $u_p(M)$ a total of
\[\begin{array}{lcl}
2L\left(\sum_{r \notin \{q,q',q''\}} (a_r-1)\right)&= &2L((\sum_r a_r) -(a_q+a_{q'}+a_{q''})-(3k-3)).
\end{array}\]
This means that for $p=1,\ldots,k$,
\[u_p(M)= 2(a_q+a_{q'}+a_{q''})+2L\left(\sum a_r\right) -6L(k-1) -3(L+1).\]
Let $x$ be the allocation with for $p=1,\ldots,k$,
$$x_p=2c+2L\left(\sum a_r\right)-6L(k-1)-3(L+1),$$
and
$$x_{k+1}=3k(L+1).$$
We now observe that there is a matching $M\in {\mathcal M}$ with $u_p(M)=x_p$  for $p=1,\ldots,k+1$  if and only if $(a_1,\ldots,a_{3k})$ is a yes-instance of $3$-{\sc Partition}. {Such a matching $M$ is also weakly close to $x$, since $\max_{p \in N}\{|x_p-u_p(M)|\} = 0$.} Moreover, as $3$-{\sc Partition} is strongly NP-complete,  $a_1, \dots, a_{3k}$ can be represented in unary.
Thus, the size of $(a_1,\ldots,a_{3k})$ is $kc$.
Hence, $(\overline{G},w)$ has polynomial size. \qed
\end{proof}

\noindent
We continue with the proof of Theorem~\ref{t-hard3}.

\medskip
\noindent
{\bf Theorem~\ref{t-hard3} (restated).}
{\it It is \NP-hard to find a {weakly close} maximum weight matching for a $3$-sparse compact directed partitioned matching game $(N,v)$ with $n=2$ and target allocation~$x$.}

\begin{proof}
We reduce again from the \NP-complete {\sc Partition} problem~\cite{GJ79}. From a given instance $(a_1,\ldots,a_k)$ of {\sc Partition} we construct a $3$-sparse compact {directed} partitioned matching game $(N,v)$ {with $n=2$}. {Let} the undirected compatibility graph $G=(V,E)$ for $(N,v)$ {be the} disjoint union $C_1+\ldots + C_k$ of $k$ cycles $C_1,\ldots,C_k$, where for $i\in \{1,\ldots,k\}$, $C_i$ has length $4a_i+{4}$.
{For each~$C_i$ we do as follows. F}or some sufficiently large integer $L$, say $L=\sum a_i$, we {assign} to each edge of ${C_i}$ weight $L$, ${L}+1$ or $L+\frac{1}{2}$ {as described below}.

Let $e$ and $\overline{e}$ be two {``opposite edges'' in cycle $C_{{i}}$, that is $e$ and $\overline{e}$ are} of maximum distance from each other. Assign weights $w_e=L$ and $w_{\overline{e}}=L+1$ to these edges. Weights $w_e$ and $w_{\overline{e}}$ are assumed to be split equally to their corresponding two opposite arcs. Removing $e$ and $\overline{e}$ splits $C_{{i}}$ into two paths $P_1$ and $P_2$ of length $2a_i+1$ each. The edge weights on these two paths alternate between $L$ and $L+1$ except for their last edge, which has weight $L+\frac{1}{2}$. More precisely, $P_1$ starts with an edge (say, incident to $e$) of weight $L+1$ and continues alternating between edges of weight $L+1$ and $L$ until its last edge (incident to $\overline{e}$) gets weight $L+ \frac{1}{2}$ (instead of $L+1$). Similarly, $P_2$ starts with an edge of weight $L$, incident to $e$, and alternates between weights $L+1$ and $L$ until the last edge gets weight $L+\frac{1}{2}$ (instead of $L$). See Figure~\ref{fig:M} for {an example,} where $a_i=5$.

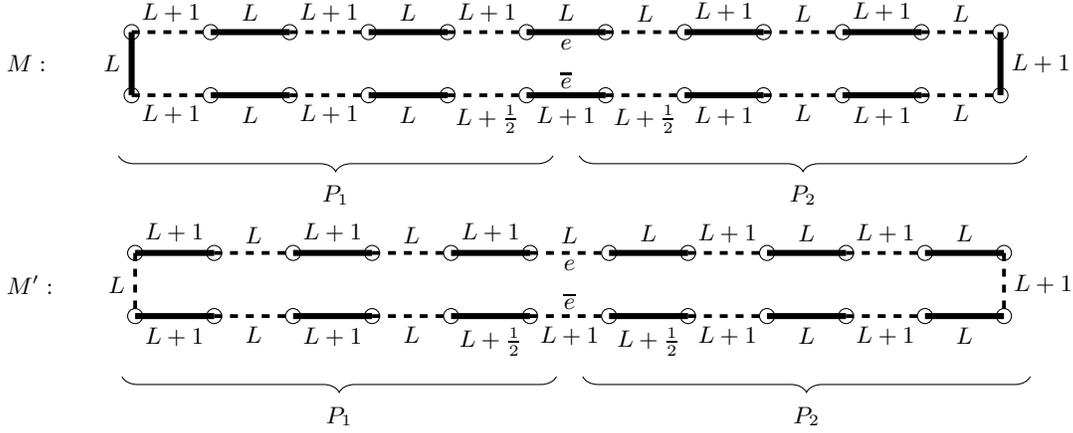
\begin{figure}
 \begin{tikzpicture}[scale=.35]
\node at (-32,1){$M:$};
\foreach \i in {1,2,3,4,5,6,7,8,9,10,11,12}
  {
    \def \x{3*\i-31}
    \node[minimum size=2mm] (\i,0)[draw=black,circle, inner sep=0pt] at (\x,-0.2){};
    \node[minimum size=2mm] (\i,2)[draw=black,circle, inner sep=0pt] at (\x,2.2){};

  }
\draw[dashed, line width=.5mm] (-28,-0.2)--(-25,-0.2) node[midway,below]{\footnotesize{$L+1$}};
\draw[dashed, line width=.5mm] (-22,-0.2)--(-19,-0.2) node[midway,below]{\footnotesize{$L+1$}};
\draw[dashed, line width=.5mm] (-16,-0.2)--(-13,-0.2) node[midway,below]{\footnotesize{$L+\frac{1}{2}$}};
\draw[dashed, line width=.5mm] (-10,-0.2)--(-7,-0.2) node[midway,below]{\footnotesize{$L+\frac{1}{2}$}};
\draw[dashed, line width=.5mm] (-4,-0.2)--(-1,-0.2) node[midway,below]{\footnotesize{$L$}};
\draw[dashed, line width=.5mm] (2,-0.2)--(5,-0.2) node[midway,below]{\footnotesize{$L$}};
\draw[line width=.8mm] (-25,-0.2)--(-22,-0.2) node[midway,below]{\footnotesize{$L$}};
\draw[line width=.8mm] (-19,-0.2)--(-16,-0.2) node[midway,below]{\footnotesize{$L$}};
\draw[line width=.8mm] (-13,-0.2)--(-10,-0.2) node[midway,below]{\footnotesize{$L+1$}};
\draw[line width=.8mm] (-7,-0.2)--(-4,-0.2) node[midway,below]{\footnotesize{$L+1$}};
\draw[line width=.8mm] (-1,-0.2)--(2,-0.2) node[midway,below]{\footnotesize{$L+1$}};

\node at (-11.5,0.375){{$\overline{e}$}};
\node at (-11.5,1.75){{$e$}};

\draw[dashed, line width=.5mm] (-28,2.2)--(-25,2.2) node[midway,above]{\footnotesize{$L+1$}};
\draw[dashed, line width=.5mm] (-22,2.2)--(-19,2.2) node[midway,above]{\footnotesize{$L+1$}};
\draw[dashed, line width=.5mm] (-16,2.2)--(-13,2.2) node[midway,above]{\footnotesize{$L+1$}};
\draw[dashed, line width=.5mm] (-10,2.2)--(-7,2.2) node[midway,above]{\footnotesize{$L$}};
\draw[dashed, line width=.5mm] (-4,2.2)--(-1,2.2) node[midway,above]{\footnotesize{$L$}};
\draw[dashed, line width=.5mm] (2,2.2)--(5,2.2) node[midway,above]{\footnotesize{$L$}};
\draw[line width=.8mm] (-25,2.2)--(-22,2.2) node[midway,above]{\footnotesize{$L$}};
\draw[line width=.8mm] (-19,2.2)--(-16,2.2) node[midway,above]{\footnotesize{$L$}};
\draw[line width=.8mm] (-13,2.2)--(-10,2.2) node[midway,above]{\footnotesize{$L$}};
\draw[line width=.8mm] (-7,2.2)--(-4,2.2) node[midway,above]{\footnotesize{$L+1$}};
\draw[line width=.8mm] (-1,2.2)--(2,2.2) node[midway,above]{\footnotesize{$L+1$}};

\draw[line width=.8mm] (5,2.2)--(5,-0.2) node[midway,right]{\footnotesize{$L+1$}};
\draw[line width=.8mm] (-28,2.2)--(-28,-0.2) node[midway,left]{\footnotesize{$L$}};

\draw[decorate,decoration={brace, amplitude=6}]  (-12,-2.5) --   (-28.5,-2.5);
\node at (-20.25,-4){{$P_1$}};

\draw[decorate,decoration={brace, amplitude=6}]  (6,-2.5) --   (-11,-2.5);
\node at (-2.5,-4){{$P_2$}};
\end{tikzpicture}

\begin{tikzpicture}[scale=0.35]
\node at (-32,1){$M':$};
\foreach \i in {1,2,3,4,5,6,7,8,9,10,11,12}
  {
    \def \x{3*\i-31}
    \node[minimum size=2mm] (\i,0)[draw=black,circle, inner sep=0pt] at (\x,-0.2){};
    \node[minimum size=2mm] (\i,2)[draw=black,circle, inner sep=0pt] at (\x,2.2){};

  }
\draw[line width=.8mm] (-28,-0.2)--(-25,-0.2) node[midway,below]{\footnotesize{$L+1$}};
\draw[line width=.8mm] (-22,-0.2)--(-19,-0.2) node[midway,below]{\footnotesize{$L+1$}};
\draw[line width=.8mm] (-16,-0.2)--(-13,-0.2) node[midway,below]{\footnotesize{$L+\frac{1}{2}$}};
\draw[line width=.8mm] (-10,-0.2)--(-7,-0.2) node[midway,below]{\footnotesize{$L+\frac{1}{2}$}};
\draw[line width=.8mm] (-4,-0.2)--(-1,-0.2) node[midway,below]{\footnotesize{$L$}};
\draw[line width=.8mm] (2,-0.2)--(5,-0.2) node[midway,below]{\footnotesize{$L$}};
\draw[dashed, line width=.5mm] (-25,-0.2)--(-22,-0.2) node[midway,below]{\footnotesize{$L$}};
\draw[dashed, line width=.5mm] (-19,-0.2)--(-16,-0.2) node[midway,below]{\footnotesize{$L$}};
\draw[dashed, line width=.5mm] (-13,-0.2)--(-10,-0.2) node[midway,below]{\footnotesize{$L+1$}};
\draw[dashed, line width=.5mm] (-7,-0.2)--(-4,-0.2) node[midway,below]{\footnotesize{$L+1$}};
\draw[dashed, line width=.5mm] (-1,-0.2)--(2,-0.2) node[midway,below]{\footnotesize{$L+1$}};

\node at (-11.5,0.375){{$\overline{e}$}};
\node at (-11.5,1.75){{$e$}};

\draw[line width=.8mm] (-28,2.2)--(-25,2.2) node[midway,above]{\footnotesize{$L+1$}};
\draw[line width=.8mm] (-22,2.2)--(-19,2.2) node[midway,above]{\footnotesize{$L+1$}};
\draw[line width=.8mm] (-16,2.2)--(-13,2.2) node[midway,above]{\footnotesize{$L+1$}};
\draw[line width=.8mm] (-10,2.2)--(-7,2.2) node[midway,above]{\footnotesize{$L$}};
\draw[line width=.8mm] (-4,2.2)--(-1,2.2) node[midway,above]{\footnotesize{$L$}};
\draw[line width=.8mm] (2,2.2)--(5,2.2) node[midway,above]{\footnotesize{$L$}};
\draw[dashed, line width=.5mm] (-25,2.2)--(-22,2.2) node[midway,above]{\footnotesize{$L$}};
\draw[dashed, line width=.5mm] (-19,2.2)--(-16,2.2) node[midway,above]{\footnotesize{$L$}};
\draw[dashed, line width=.5mm] (-13,2.2)--(-10,2.2) node[midway,above]{\footnotesize{$L$}};
\draw[dashed, line width=.5mm] (-7,2.2)--(-4,2.2) node[midway,above]{\footnotesize{$L+1$}};
\draw[dashed, line width=.5mm] (-1,2.2)--(2,2.2) node[midway,above]{\footnotesize{$L+1$}};

\draw[dashed, line width=.5mm] (5,2.2)--(5,-0.2) node[midway,right]{\footnotesize{$L+1$}};
\draw[dashed, line width=.5mm] (-28,2.2)--(-28,-0.2) node[midway,left]{\footnotesize{$L$}};

\draw[decorate,decoration={brace, amplitude=6}]  (-12,-2.5) --   (-28.5,-2.5);
\node at (-20.25,-4){{$P_1$}};

\draw[decorate,decoration={brace, amplitude=6}]  (6,-2.5) --   (-11,-2.5);
\node at (-2.5,-4){{$P_2$}};

\end{tikzpicture}
\caption{{The matchings $M$ (top) and $M'$ (bottom) in $C_i$ for ${a_i=5}$ with edges ${e}$ and ${\overline{e}}$ in the middle.}}\label{fig:M}
 \end{figure}

 We let $U_1$ and $U_2$ denote the vertex sets of $P_1$ and $P_2$, respectively.
 {We note that}
$C_{{i}}$ has exactly two
 maximum weight matchings, namely its two complementary perfect matchings $M$ and $M'$, where
$M$ is the perfect matching that matches both $e$ and $\overline{e}$ and $M'$ is the complement of~$M${; see also Figure~\ref{fig:M}}.
{Let $V_1$ be the union of all the $U_1$s in each $C_i$ and $V_2$ be the union of all the $U_2$s in each $C_i$.}

We compute:
 \[
  u_1(M) = \frac{1}{2}L+\frac{1}{2}(L+1)+a_iL
  =  L(a_i+1)+\frac{1}{2}
  \]
 and
 \[
  u_2(M)
  =
  \frac{1}{2}L+\frac{1}{2}(L+1)+a_i(L+1)
  = L(a_i+1)+\frac{1}{2}+a_i,
\]
while $u_1(M')
=
L(a_i+1)+\frac{1}{2}+a_i$, and
  $u_2(M')
  =
  L(a_i+1)+\frac{1}{2}$.

Recall that we have $k$ components $C_i$, each with two complementary maximum weight (perfect) matchings. So in the graph~$G$ consisting of these $k$ components $C_i$ we have $2^k$ maximum weight matchings, obtained by picking one of the two complementary $M$ and $M'$ in each $C_i$. Consider the allocation~$x$ with $$x_1= x_2=L\left(\sum a_i+1\right) +\frac{1}{2}\sum a_i +k/2,$$ and assume these can be realized by a suitable maximum weight matching {$M_G$ in $G$}. Let $I \subseteq \{1, \dots, k\}$ be the set of indices $i$ such that ${M_G}$ picks $M$ in  $C_i$. With respect to {$M_G$, we find that} $V_1$ has utility $L\sum (a_i+1) +k/2+\sum_Ia_i$. Such a matching {$M_G$} exists if and only if $(a_1,\ldots,a_k)$ is a yes-instance of {\sc Partition}. {Moreover, $M_G$ is weakly close to $x$, as $|x_1-u_1(M_G)|=|x_2-u_2(M_G)|=0$.}
This completes the reduction.

Each component $C_i$ of the graph we construct has a description of length $O(\log (ka_{max}))$, where $a_{max}$ denotes the maximum $a_i$; note that $L$ is bounded by $\log (k a_{max})$ and the length of $C_i$ is bounded by  $a_i$. Hence,  allowing compact descriptions,  the weighted graph we constructed has size $O(k \log(k a_{max}))$, which is polynomial in the size of $(a_1,\ldots,a_k)$. \qed
\end{proof}

\noindent
Finally, we prove Theorem~\ref{t-hard4}.

\medskip
\noindent
{\bf Theorem~\ref{t-hard4} (restated).}
{\it {\sc Exact Perfect Matching} and the problem of finding a {weakly close} maximum weight matching for a $3$-sparse perfect directed partitioned matching game with $n=2${,} and target allocation~$x$ are polynomially equivalent.}

\begin{proof}
First suppose that we can solve {\sc Exact Perfect Matching} in polynomial time. Let $(N,v)$ be a {$3$-sparse} perfect directed partitioned matching game {with $n=2$} defined on $(\overline{G},w)$ with partition $(V_1,V_2)$.
{Recall that for every $2$-cycle $iji$ with $i\in V_1$ and $j\in V_2$, we set $w_{ij}=\frac{1}{3}$ and $w_{ji}=\frac{2}{3}$, and
for every other $2$-cycle $iji$ in $\overline{G}$ (so where $i,j\in V_1$ or $i,j\in V_2$), we set $w_{ij}=w_{ji}=\frac{1}{2}$.}
Recall {also} that we denote the underlying undirected graph corresponding to $\overline{G}=(V,A)$ by $G=(V,E)${, which we assume has a perfect matching.}

Let $(x_1,x_2)$ be an allocation. As $G$ has a perfect matching by definition, we find that $x_1+x_2={\frac{1}{2}|V|}$. We need to check if there exists a matching $M\in {\mathcal M}$ (which will be perfect) with $u_1(M)\in I_1=[x_1-\delta,x_1+\delta]$ and $u_2(M)\in I_2=[x_2-\delta,x_2+\delta]$ for some given $\delta\geq 0$.
{If we can do this, then we can do a search on $\delta$ similar to the one we employed in the {\tt Lex-Min} algorithm, leading to a maximum weight matching that is weakly close to~$x$. This can be done in polynomial time, as we still have a
a polynomial number of values for $|x_1-u_1(M)|$ and  $|x_2-u_2(M)|$ due to the $3$-sparsity. Hence, we can follow the reasoning used in the proof of Theorem~\ref{t-easy1} for
 analyzing the running time of {\tt Lex-Min}.}

Colour all edges of $G$ with one end-vertex in $V_1$ and the other one in $V_2$ red. This gives us the set~$R$. Colour all remaining edges blue, that is, let $B=E\setminus R$. We check for $k=1,\ldots,|R|$ whether there exists a perfect matching of $G$ with exactly $k$ red edges. This takes polynomial time by our assumption on {\sc Exact Perfect Matching}. Each time we find a solution $M$ we let $\ell_i$ be the number of (blue) edges with both end-vertices in $V_i$ for $i=1,2$, and we check whether $\frac{2}{3}k+\ell_1$ belongs to $I_1$ and $\frac{1}{3}k+\ell_2$ belongs to $I_2$ (note that if $k$ is fixed, then $\ell_1$ and $\ell_2$ are fixed as well).

\medskip
\noindent
Now suppose that we can find in polynomial time a {weakly close} maximum weight matching for a perfect directed partitioned matching game and target allocation~$x$.
Let $G=(V,E)$ be an undirected graph with a partition $(R,B)$ of $E$ into red and blue edges forming, together with an integer~$k\geq 0$, an instance of {\sc Exact Perfect Matching}.

We subdivide each edge of $G$ twice, that is, we replace each edge $e=ij$ with vertices $i',j'$ and edges $ii'$, $i'j'$, $j'j$.  Let $G'=(V',E')$ be the resulting graph, so $V'\setminus V$ is the set of the $2|E|$ new vertices, which we call {\it subdivision vertices}. Any perfect matching $M$ in $G$ translates into a unique perfect matching $M'$ in $G'$, and vice versa. Namely, for every $ij\in E$, we have $ij\in M$ if and only if $ii'$, $j'j\in M'$, and also $ij\notin M$ if and only if $i'j'\in M'$. We call $M'$ the {\it transform} of $M$.
We let $V_1'$ be the set of the subdivision vertices on red edges in $G$, and we let $V_2' = V' \setminus V_1'$.
We let $R'$ denote the edges with one end-vertex in $V_1'$ and the other end-vertex in $V_2'$.
Then the transform $M'$ of a perfect matching $M$ with exactly $k$ edges in $R$ has $2k$ edges in $R'$, and vice versa.
To solve the latter
problem, we
define a perfect partitioned matching game corresponding to $G'$ and $(V_1',V_2')$ and we
choose $(x_1,x_2)=({|R| + \frac{1}{3}k},{\frac{1}{2}|V| + |B| - \frac{1}{3}k})$ as allocation.
{Note that the size of a perfect matching in $G'$ is $\frac{1}{2}|V|+|E|=\frac{1}{2}|V|+|R|+|B|=x_1+x_2$.}

We now check if there exists a matching $M{'}\in {\cal M'}$ (the set of perfect matchings of $G'$) such that $u_1(M{'})={|R| + \frac{1}{3}k}$ and $u_2(M{'})={\frac{1}{2}|V| + |B| - \frac{1}{3}k}$. {If such a matching $M'$ exists, it is also weakly close to $x$, since
$|x_1-u_1(M')|=|x_2-u_2(M')|=0$.}
By our assumption we can do this in polynomial time.\qed
\end{proof}

\section{Conclusions}\label{s-con}

We introduced a new class of cooperative games: partitioned matching games. We showed how we can use partitioned matching games to model international kidney exchange programmes. These programmes are seen as the next step in the medical field of organ transplantations~\cite{Bo_etal17,Va_etal19}. We provided the theoretical basis for this application by proving a number of computational complexity results for partitioned matching games.

We found two sets of results. One set of results was about ensuring stability of the international collaboration. The aim was to choose in each round of the international programme a kidney transplant distribution as close as possible to some prescribed fair distribution for that round. Roughly speaking, we proved that this problem can be solved efficiently when transplant weights are equal, but otherwise the problem becomes quickly computationally hard. We pose the following two open problems; recall that we showed some partial results
{for directed partitioned matching games $(N,v)$ with $|N|=2$}
in Theorems~\ref{t-hard3} and~\ref{t-hard4}.

\begin{enumerate}
\item Are there constants $c$ and $d$ such that the problem of finding a {weakly close} maximum weight matching for a
$d$-sparse partitioned matching game $(N,v)$ with width~$c$ and target allocation~$x$ is \NP-hard?\\[-8pt]
\item Are there constants $n$ and $d$ such that the problem of finding a {weakly close} maximum weight matching for a
$d$-sparse partitioned matching game $(N,v)$ with $|N|=n$ and target allocation~$x$ is \NP-hard?
\end{enumerate}

\noindent
In our other set of results, we linked the core of partitioned matching games to the core of $b$-matching games. We resolved a complexity gap for computing core allocations of $b$-matching games. As a consequence we can settle the computational complexity for the three core-related problems P1--P3 for partitioned matching games as well{;} see also Table~\ref{t-tabtab}.
{We note that Table~\ref{t-tabtab} contains two co-\NP-hardness results for P2. We do not know if P2 for $b$-matching games with $b\not\leq 2$ and for partitioned matching games of width $c\geq 3$ is co-\NP-complete. Given that P1 is co-\NP-complete for these cases (see Table~\ref{t-tabtab}), P2 is likely to be computationally harder. We leave this question for future research.}

It is also interesting to consider other solution concepts for $b$-matching games and partitioned matching games, such as the nucleolus. K\"onemann, Pashkovich and Toth~\cite{KPT20} proved that the nucleolus of a matching game can be computed in polynomial time.  In contrast, K\"onemann, Toth and Zhou~\cite{KTZ21} proved that computing the nucleolus is \NP-hard even for uniform $b$-assignment games with $b\leq 3$.  We refer to~\cite{BHIM10,KT20,KTZ21} for some positive results (see also~\cite{BBJPX}), but determining the complexity of computing the nucleolus is still open for the following games (see also~\cite{KTZ21}):

\begin{enumerate}
\item $b$-matching games with $b\leq 2$,
\item partitioned matching games,
\item partitioned matching games  with width $c\leq 2$, and
\item partitioned matching games with width $c\leq 3$.
\end{enumerate}

\noindent
{\it Acknowledgments.} We thank two anonymous reviewers for helpful comments that improved the readability of our paper.

\section*{Statements and Declarations}

\subsection*{Conflict of interest} The authors declared that they have no conflict of interest.

\end{document}